\documentclass[acmsmall,screen]{acmart}

\AtBeginDocument{%
  }

\copyrightyear{2024}
\acmYear{2024}
\acmDOI{10.1145/3688824}

\setcopyright{rightsretained}
\acmJournal{ToCT}
\acmYear{2024} \acmVolume{16} \acmNumber{4} \acmArticle{21} \acmMonth{11}\acmDOI{10.1145/3688824}

\usepackage[capitalize,nameinlink,noabbrev]{cleveref}
\usepackage{amsmath}

\newcommand{\BHM}{\ensuremath{\mathsf{BHM}}}

\newcommand{\HH}{\ensuremath{\mathsf{HH}}}
\newcommand{\YES}{\ensuremath{\mathsf{YES}}}
\newcommand{\Y}{\ensuremath{\mathcal{Y}}}
\newcommand{\N}{\ensuremath{\mathcal{N}}}
\newcommand{\NO}{\ensuremath{\mathsf{NO}}}
\newcommand{\Z}{\ensuremath{\mathbb{Z}}}
\newcommand{\R}{\ensuremath{\mathbb{R}}}
\newcommand{\A}{\ensuremath{\mathcal{A}}}
\newcommand{\Exp}{\mathbb{E}}
\DeclareMathOperator{\poly}{poly}
\newcommand{\ket}[1]{|#1\rangle}

\newtheorem{theorem}{Theorem}
\newtheorem{definition}[theorem]{Definition}
\newtheorem{fact}[theorem]{Fact}
\newtheorem{lemma}[theorem]{Lemma}

\newtheorem{conjecture}[theorem]{Conjecture}

\newtheorem{remark}[theorem]{Remark}
\newtheorem{result}{Result}

\DeclareRobustCommand{\DE}[2]{#2}
\begin{document}

\title{Matrix hypercontractivity, streaming algorithms and LDCs: the large alphabet case}

\author{Srinivasan Arunachalam}

\email{Srinivasan.Arunachalam@ibm.com}
\orcid{0000-0001-6014-6624}
\affiliation{%
  \institution{IBM Quantum, IBM T.J. Watson Research Center}
  \city{Yorktown Heights}
  \state{New York}
  \country{USA}
}

\author{Joao F. Doriguello}
\orcid{0000-0002-8265-7334}
\email{doriguello@renyi.hu}
\affiliation{%
  \institution{HUN-REN Alfr\'ed R\'enyi Institute of Mathematics}
  \city{Budapest}
  \country{Hungary}}
\affiliation{%
  \institution{Centre for Quantum Technologies, National University of Singapore}
  \country{Singapore}}

\authorsaddresses{%
Srinivasan Arunachalam and Joao F. Doriguello contributed equally to this research.\\
JFD was supported by the National Research Foundation, Singapore and A*STAR under the CQT Bridging Grant and the
Quantum Engineering Programme Award number NRF2021-QEP2-02-P05 and by ERC Grant No. 810115-DYNASNET.
Authors’ Contact Information: Srinivasan Arunachalam, IBM, Yorktown Heights, New York, USA; e-mail: srinivasan.
arunachalam@ibm.com; Joao F. Doriguello, Alfréd Rényi Institute of Mathematics, Budapest, Hungary and Centre for
Quantum Technologies, National University of Singapore, Singapore; e-mail: doriguello@renyi.hu.
}


\begin{abstract}
    We prove a hypercontractive inequality for matrix-valued functions defined over large alphabets. {In order to do so}, we prove a generalization of the powerful $2$-uniform convexity inequality for trace norms of Ball, Carlen, Lieb (Inventiones Mathematicae'94). Using our hypercontractive~inequality, we present upper and lower bounds for the communication complexity of the Hidden Hypermatching problem defined over large alphabets. We then consider streaming algorithms for approximating the value of Unique Games on a hypergraph with $t$-size hyperedges. By using our communication lower bound, we show that every streaming algorithm in the adversarial model achieving an $(r-\varepsilon)$-approximation of this value requires $\Omega(n^{1-2/t})$ quantum space, {where $r$ is the alphabet size}. We next present a lower bound for locally decodable codes ($\mathsf{LDC}$) $\mathbb{Z}_r^n\to \mathbb{Z}_r^N$ over large alphabets with recoverability probability at least $1/r + \varepsilon$. Using hypercontractivity, we give an exponential lower bound $N= 2^{\Omega(\varepsilon^4 n/r^4)}$ for $2$-query (possibly non-linear) $\mathsf{LDC}$s over $\mathbb{Z}_r$ and using the non-commutative Khintchine inequality we prove an improved lower bound of $N= 2^{\Omega(\varepsilon^2 n/r^2)}$.
\end{abstract}

\ccsdesc[500]{Theory of computation~Streaming models}
\ccsdesc[500]{Theory of computation~Quantum communication complexity}
\ccsdesc[500]{Theory of computation~Communication complexity}
\ccsdesc[500]{Theory of computation~Error-correcting codes}

\keywords{Streaming Algorithms, Locally Decodable Codes, Quantum Communication Complexity, Hypercontractivity, Max-Cut, Unique Games}

\received{13 June 2023}
\received[revised]{22 May 2024}
\received[accepted]{02 July 2024}

\maketitle

\section{Introduction}
In this paper we prove new results in two areas of theoretical computer science that have received a lot of attention recently: \emph{streaming algorithms} and \emph{locally decodable codes}.

\emph{Streaming algorithms} is a model of computation introduced by Alon, Matias, and Szegedy~\cite{alon1999space} (for which they won the G{\"o}del Prize in 2005) in order to understand the space complexity of approximation algorithms to solve problems. In the last decade, there have been several results in the direction of proving upper and lower bounds for streaming algorithms for combinatorial optimization problems~\cite{verbin2011streaming,goel2012communication,kapralov2014streaming,kapralov20171,kapralov2019optimal,DBLP:conf/approx/GuruswamiT19,DBLP:conf/focs/ChouGV20,chou2021approximabilityB,assadi2021simple,chen2021almost}. 
The goal here is to obtain a $1/\gamma$-approximation (for some $\gamma \leq 1$) of the optimum value of the combinatorial optimization problem with as little \emph{space} as possible.
One favourite problem considered by many works is the well-known \emph{Max-Cut {problem}}, or its generalization over large alphabets $\Z_r$, \emph{Unique Games}. {There is a $2$-approximation algorithm for Max-Cut on $n$ vertices using logarithmic space}: one simply counts the number of edges in the graph (which requires only a counter of size $2\log n$) and outputs half this count. Moreover, one can obtain a graph sparsifier using $O(n/\varepsilon^2)$ space~\cite{ahn2012graph} and, from it, a $(1+\varepsilon)$-approximation for the Max-Cut value. On the other hand, a sequence of works~\cite{kapralov2014streaming,kapralov20171,kapralov2019optimal,chou2022linear} showed that getting a $(2-\varepsilon)$-approximation requires linear space. A similar, but less optimized, {phenomenon} was observed for Unique Games, i.e., there is a threshold behaviour in complexity going from $r$ to $(r-\varepsilon)$-approximation, {where $r$ is the alphabet size}. Curiously, many of these lower bounds were proven via variants of a problem called \emph{Boolean Hidden Matching} $(\BHM$) and it is well known that $\BHM$ can be solved using logarithmic \emph{quantum space}, so a natural question is, could quantum space help {solve} these combinatorial optimization problems? One corollary from~\cite{kapralov2014streaming,shi2012limits} is that obtaining the \emph{strong} $(1+\varepsilon)$-approximation factor for Max-Cut and Unique Games streaming algorithms is quantum-hard. However, understanding the space complexity of streaming in the widely-studied, weaker regime of $(2-\varepsilon)$-approximation (for Max-Cut) or $(r-\varepsilon)$-approximation (for Unique Games over $\Z_r$), it is still unclear whether there could be any savings in the quantum regime.

\emph{Locally decodable codes} ($\mathsf{LDC}$s) are error correcting codes $C:\Sigma^n\rightarrow\Gamma^N$ (for alphabets $\Sigma,\Gamma$) that allow transmission of information over noisy channels. By querying a few locations of a noisy codeword $\widetilde{C}(x)$, one needs to reconstruct an arbitrary coordinate of $x\in \Sigma^n$ with probability at least $1/|\Sigma|+\varepsilon$. The main goal in this field is to understand trade-offs between $N$ and $n$. 
$\mathsf{LDC}$s have found several applications in pseudorandom generators, hardness
amplification, private information retrieval schemes, cryptography, {and} complexity theory ({we} refer to~\cite{yekhanin2012locally,gopi2018locality} for detailed expositions). Despite their ubiquity, $\mathsf{LDC}$s are not well understood, even with the simplest case of \emph{$2$-query $\mathsf{LDC}$s}. 
For the case when $\Sigma=\Gamma=\{0,1\}$, exponential lower bounds of $N=2^{\Omega(n)}$  were established over two decades back~\cite{goldreich2002lower,DBLP:journals/jcss/KerenidisW04,dvir2007locally}. In contrast, a breakthrough result of Dvir and Gopi~\cite{dvir20162} in 2015 showed how to construct $2$-query $\mathsf{LDC}$s with \emph{subexponential} length in the regime when $\Sigma=\{0,1\}$ and $\Gamma$ is a finite field {of size $2^{n^{o(1)}}$}. Despite these results, our knowledge of such $N$ and $n$ trade-offs for $2$-query $\mathsf{LDC}$s is still lacking, specially for the not very well studied case when $\Sigma=\Gamma=\Z_r$.\footnote{{We shall think of an $r$-size alphabet set as the ring $\mathbb{Z}_r$ since its structure will be extensively used in our analyses.}} {What is the relation between $n$ and $N$ in this case?}

Prior works that handled simpler versions of the questions above used one technical tool successfully: \emph{hypercontractivity} for real-valued functions over the Boolean cube. Since we are concerned with proving quantum lower bounds for streaming algorithms and $\mathsf{LDC}$s when the input alphabet is $\Z_r$, it leads us to the following main question: 
    \emph{Is there a version of hypercontractivity for matrix-valued functions over~$\Z_r$?}

\subsection{Our results} 

Summarizing our main contributions, we first prove a version of hypercontractivity for matrix-valued functions $f:\Z_r^n\rightarrow\mathbb{C}^{m\times m}$. The proof of this crucially relies on proving {a generalization of} the powerful $2$-uniform convexity by Ball, Carlen, and Lieb~\cite{ball1994sharp}. Using this new hypercontractivity theorem, we prove a quantum space lower bound for streaming algorithms.  It is easy to see that obtaining a $2$-approximation algorithm for Max-$t$-Cut on $n$ vertices and $t$-sized hyperedges in the classical streaming model can be done in $O(\log n)$ space, and we show that obtaining a $1.99$-approximation algorithm in the adversarial model requires $\Omega(n^{1-2/t})$ quantum space or $\Omega(n^{1-1/t})$ classical space. As far as we are aware, this is the first quantum space lower bound for an optimization problem. {Compared to the prior works of Kapralov, Khanna, and Sudan~\cite{kapralov2014streaming} and of Kapralov and Krachun~\cite{kapralov2019optimal}, who obtained an $\Omega(\sqrt{n})$ and $\Omega(n)$ classical space lower bounds in the \emph{random} and \emph{adversarial} models, respectively, for $(2-\varepsilon)$-approximation algorithms, our lower bounds apply to the adversarial model and are derived through significantly simpler proofs.} We further generalize our results to the case of hypergraphs {with $t$-size hyperedges} and vertices taking values over $\Z_r$. These hypergraphs can naturally be viewed as instances of Unique Games wherein the constraints are over $\Z_r$. Here again, we prove that obtaining an $r$-approximation algorithm requires $O(\log n)$ classical space while a $(r-\varepsilon)$-approximation algorithm requires 
$\Omega(n^{1-2/t})$ quantum space {(hiding the dependence on $r$ and $\varepsilon$)}.

We next show an $N = 2^{\Omega(n/r^2)}$ lower bound for (even non-linear) $\mathsf{LDC}$s over $\Z_r$. In particular, for all $r$ smaller than $\sqrt{n}$, we prove an exponential in $n$ lower bound for $\mathsf{LDC}$s over $\Z_r$. Previous main results in this direction were by Goldreich et al.~\cite{goldreich2002lower} for $r=2$ and linear $\mathsf{LDC}$s, Kerenidis and de Wolf~\cite{DBLP:journals/jcss/KerenidisW04} for $r=2$ and \emph{non-linear} $\mathsf{LDC}$s, Wehner and de Wolf~\cite{wehner2005improved} for non-linear $\mathsf{LDC}$s from $\{0,1\}^n\to\Z_r^N$ and finally by Dvir and Shpilka~\cite{dvir2007locally} for $r>2$ but linear $\mathsf{LDC}$s. Apart from the result of~\cite{dvir2007locally}, we are not aware of any lower bounds for non-linear $\mathsf{LDC}$s  from $\Z_r^n\to\Z_r^N$, even though it is a very natural question with connections to other fundamental problems, such as private information retrieval~\cite{katz2000efficiency,goldreich2002lower}, additive combinatorics~\cite{briet2016outlaw}, and quantum complexity theory~\cite{aaronson2018pdqp}, to cite a few. Moreover, the setting of equal alphabets $\Sigma = \Gamma = \mathbb{Z}_r$ is particularly relevant and natural in the context of polynomial identity testing and $\Sigma\Pi\Sigma$ circuits, which has been previously considered by Dvir and Shpilka~\cite{dvir2007locally} for linear $\mathsf{LDC}$s. Furthermore, we are not aware of a formal reduction between $\mathsf{LDC}$s with $\Sigma=\{0,1\}$ and $\Sigma=\Z_r$, specially with recovery probability $1/|\Sigma| + \varepsilon$.
Finally, some past works define $\mathsf{LDC}$s over general $\Sigma$ with success probability $\geq \operatorname{Pr}[\text{wrong output}] + \varepsilon$~\cite{gopi2018locality}, $\geq 1/2 + \varepsilon$~\cite{goldreich2002lower} or $\geq 1-\varepsilon$~\cite{dvir2011matrix}. These alternative definitions are encompassed by ours by considering $\varepsilon$ a constant large enough. In the remaining part of the introduction, we describe these contributions in more detail.

\subsection{Matrix hypercontractive inequality (over large alphabets)}
\paragraph{Fourier analysis on the Boolean cube.} We first discuss the basics of Fourier analysis before stating our result. Let $f:\{0,1\}^n\rightarrow \R$ be a function, then the Fourier decomposition of $f$ is
$$
f(x)=\sum_{S\in\{0,1\}^n} \widehat{f}(S)(-1)^{S\cdot x},
$$
where $S\cdot x = \sum_{i=1}^n S_ix_i\in\{0,1\}$ and the \emph{Fourier coefficients} $\widehat{f}(S)$ of $f$ are $\widehat{f}(S)=\Exp_x[f(x) (-1)^{S\cdot x}]$, the expectation taken over uniformly random $x\in\{0,1\}^n$. One of the technical tools in the area of theoretical computer science is the hypercontractivity theorem proven by Bonami and Beckner~\cite{bonami1970etude,beckner1975inequalities}. In order to understand the hypercontractivity theorem, we first need to define the noise operator: for a noise parameter $\rho \in [-1,1]$, let $\operatorname{T}_\rho$ be {the} operator on the space of functions $f:\{0,1\}^n\rightarrow \R$ defined as 
$$
    ({\operatorname{T}}_\rho f)(x)=\operatorname*{\Exp}_{y\sim \mathcal{N}_\rho(x)}[f(y)],
$$
where $y\sim \mathcal{N}_\rho(x)$ denotes that the random string $y\in\{0,1\}^n$ is drawn as $y_i=x_i$ with probability $\frac{1}{2} + \frac{1}{2}\rho$ and as $y_i=x_i\oplus 1$ with probability $\frac{1}{2} - \frac{1}{2}\rho$ for each $i\in[n]$ independently. One can show that the Fourier expansion of ${\operatorname{T}}_\rho f$ can be written as {(see, e.g.,~\cite[Proposition 2.47]{o2014analysis})}
$$
({\operatorname{T}}_\rho f)(x)=\sum_{S\in\{0,1\}^n}\rho^{|S|}\widehat{f}(S)(-1)^{S\cdot x}.
$$
One way to intuitively view this expression is that ``large-weight'' Fourier coefficients are reduced by an exponential factor while ``small-weight'' Fourier coefficients remain approximately the same. {Another property of the noise operator is that it is a contraction, i.e.,} $\|{\operatorname{T}}_\rho f\|_p\leq \|f\|_p$ for every $p\geq 1$, where $\|f\|_p\triangleq\big(\Exp_x[|f(x)|^p]\big)^{1/p}$ is the standard normalized $p$-norm of the function $f$. The main hypercontractivity theorem states that this inequality holds true even if we increase the left-hand size by a larger norm (meaning that norms under the noise operator are not just contractive, but \emph{hypercontractive}), i.e., for every $p\in [1,2]$ and $\rho \leq \sqrt{p-1}$, we have that $\|{\operatorname{T}}_\rho f\|_2\leq \|f\|_p$,\footnote{The hypercontractivity theorem can be stated for arbitrary $1\leq p \leq q$ and $\rho \leq \sqrt{(p-1)/(q-1)}$, here we state it for $q=2$ since we will be concerned with this setting.} which can alternatively be written as
\begin{align}
\label{eq:basichypercontractivity}
    \left(\sum_{S\in \{0,1\}^n} \rho^{2|S|} \widehat{f}(S)^2\right)^{1/2} \leq \left(\frac{1}{2^n}\sum_{x\in\{0,1\}^n} |f(x)|^p\right)^{1/p}.
\end{align}
This inequality has found several applications in approximation theory~\cite{khot2007optimal,dinur2005hardness}, expander graphs~\cite[Section~11.5]{hoory2006expander}, circuit complexity~\cite{linial1993constant}, coding theory~\cite{carlen1993optimal}, quantum computing~\cite{gavinsky2007exponential,DBLP:journals/qic/Montanaro11} (for more applications we refer the reader to~\cite{de2008brief,o2014analysis,montanaro2012some}). All these applications deal with understanding the effect of noise on real-valued functions on the Boolean cube.

\noindent \textbf{Generalizations of hypercontractivity.} There are two natural generalizations of hypercontractivity: $(i)$ a hypercontractivity statement for arbitrary product probability spaces. In this direction, it is possible to prove a similar hypercontractive inequality: for every $p\in [1,2]$ and $f\in L^2(\Omega_1\times\cdots\times\Omega_n, \pi_1\otimes\cdots\otimes\pi_n)$, we have
\begin{align}
\label{eq:genhyperFr}
    \|{\operatorname{T}}_\rho f\|_2 \leq \|f\|_p \hspace{1mm}\text{ for  }\hspace{1mm} \rho \leq \sqrt{p-1}\cdot\lambda^{1/p-1/2},
\end{align} 
where $\lambda$ is the smallest probability in any of the finite probability spaces $(\Omega_i,\pi_i)$ (see~\cite[Chapter~10]{o2014analysis}). As a corollary, one gets a hypercontractive inequality for $f:\Z_r^n\to\mathbb{R}$; $(ii)$ a hypercontractivity statement for matrix-valued functions $f:\{0,1\}^n\rightarrow \mathbb{C}^{m\times m}$, where the Fourier coefficients $\widehat{f}(S)=\Exp_x[f(x)(-1)^{S\cdot x}]$ are now $m\times m$ complex matrices. This was considered by Ben-Aroya, Regev, and de Wolf~\cite{ben2008hypercontractive}, who proved a hypercontractivity statement by using the powerful inequality of Ball, Carlen, and Lieb~\cite{ball1994sharp}.

However, is there a generalization of hypercontractivity in both directions, i.e., a matrix-valued hypercontractivity for functions over $\Z_r$? This is  open as far as we are aware and is our first main technical result.
\begin{result}[\cref{thm:genmathypercontractivity}]
\label{result:hyper}
    For any $f:\Z_r^n\rightarrow \mathbb{C}^{m\times m}$, $p\in[1,2]$, and 
    $\rho\leq \sqrt{\frac{(p-1)(1-(p-1)^{r-1})}{(r-1)(2-p)}}$,
    \begin{align}
    \label{eq:resulthypereq}
        \left(\sum_{S\in \Z_r^n} \rho^{2|S|} \|\widehat{f}(S)\|_p^2\right)^{1/2} \leq \left(\frac{1}{r^n}\sum_{x\in\Z_r^n}\|f(x)\|_p^p\right)^{1/p},
    \end{align}
    where $\|M\|_p \triangleq\big(\sum_i\sigma_i(M)^p\big)^{1/p}$ is the Schatten $p$-norm defined from the singular values $\{\sigma_i(M)\}_i$ of the matrix $M$ and $|S|\triangleq|\{i\in[n]: S_i\neq 0\}|$ is the Hamming weight of $S\in\Z_r^n$.
\end{result}
The above result can be seen as an analogue of \cref{eq:basichypercontractivity} where the absolute values are replaced with Schatten norms. We now make a couple of remarks. First, when $m=1$ our result compares to the one in \cref{eq:genhyperFr} for $f:\Z_r^n\to\mathbb{R}$, but with a slightly worse $\rho$ parameter compared to the $(1/r)^{1/p-1/2}$ factor.  Second, for $r=2$ we recover the same inequality from~\cite{ben2008hypercontractive}. To this end, as in the proof of hypercontractive inequalities~\cite{o2014analysis,ben2008hypercontractive}, our result follows by induction on $n$. It so happens that the base case is the most non-trivial step in the proof. So for now, let us assume $n=1$, i.e., our goal is to prove \cref{eq:resulthypereq} for $n=1$.  We now consider  two special \emph{simple} cases of the~inequality.

\emph{(i)} $r=2$ and $\mathbb{C}^{m\times m}$ is replaced with real numbers: in this case, this is the well-known two-point inequality by Gross~\cite{gross1975logarithmic} used for understanding the Logarithmic Sobolev inequalities. A proof of this inequality can also be easily viewed from a geometric perspective. As far as we are aware, there is no generalized $r$-point inequality for $r>2$.
    
    \emph{(ii)} $r=2$ and $\mathbb{C}^{m\times m}$ are arbitrary matrices: in this case, we only need to deal with two matrices $f(0),f(1)$ and \cref{eq:resulthypereq} is exactly a powerful inequality in functional analysis, called the \emph{$2$-uniform convexity} of trace norms,\footnote{{The $2$ in $2$-uniform convexity comes from a lower bound on the modulus of convexity of the normed space $L_p$ with $p\in(1,2]$ (see~\cite{ball1994sharp} for more details) and not from our parameter $r$ in \cref{result:hyper} being equal to $2$.}}
    \begin{align}\label{eq:2_uniform_convexity}
        \left(\frac{\|X+Y\|_p^p + \|X-Y\|_p^p}{2}\right)^{2/p} \geq \|X\|_p^2 + (p-1)\|Y\|_p^2.
    \end{align}
    This inequality was first proven for certain values of $p$ by Tomczak-Jaegermann~\cite{tomczak1974moduli} before being extended for all $p\in[1,2]$ by Ball, Carlen, and Lieb~\cite{ball1994sharp} in 1994. Since then it has found several applications, e.g., an optimal hypercontractivity inequality for Fermi fields~\cite{carlen1993optimal}, regularized convex optimization~\cite{duchi2010composite}, and metric embedding~\cite{lee2004embedding,naor2016spectral}. $2$-uniform convexity can also be used to prove a variety of other inequalities, e.g., Khintchine's inequality~\cite{tomczak1974moduli,davis1984complex} and Hoeffding and Bennett-style bounds~\cite{pinelis1994optimum,howard2020time}. Moreover, the above result could be seen as a corollary of Hanner's inequality for matrices~\cite{ball1994sharp} 
    \begin{align*}
        \|X + Y\|_p^p + \|X - Y\|_p^p \geq |\|X\|_p + \|Y\|_p|^p + |\|X\|_p - \|Y\|_p|^p
    \end{align*}
    (originally proven for Lebesgue spaces $L_p$~\cite{hanner1956uniform}), but, unfortunately, Hanner's inequality for Schatten trace ideals is only proven for $1\leq p\leq 4/3$ and $p\geq 4$~\cite{ball1994sharp} (see~\cite{chayes2020matrix} for recent progress on Hanner's inequality for matrices). As far as we are aware, a generalization of \cref{eq:2_uniform_convexity} involving $r$ matrices was unknown.

One contribution in our work is the following generalization of a  result from Ball, Carlen, and Lieb~\cite{ball1994sharp} (note it also implies a generalization of the two-point inequality), which we believe may be of independent interest.
\begin{result}[\cref{eq:conjectureballgen}]
\label{res:bcl}
    Let $r\in\mathbb{Z}$, $r\geq 2$, $\omega_r \triangleq e^{2i\pi/r}$, $A_0,\ldots,A_{r-1}\in \mathbb{C}^{n\times n}$, and $p\in [1,2]$, then
    \begin{align*}
		\left(\frac{1}{r} \sum_{k=0}^{r-1}\left\|\sum_{j=0}^{r-1} \omega_r^{jk} A_j\right\|_p^p\right)^{2/p} \geq\left\|A_0 \right\|_p^2 + \frac{(p-1)(1-(p-1)^{r-1})}{(r-1)(2-p)}\sum_{k=1}^{r-1} \left\| A_k\right\|_p^2.
	\end{align*}
\end{result}
{Once \cref{res:bcl} is established}, the proof of \cref{result:hyper} is a simple induction argument on $n$: for the base case $n=1$, \cref{result:hyper} is exactly \cref{res:bcl}, and proving the induction step requires an application of Minkowski's inequality. Since this proof is similar to the one in~\cite{ben2008hypercontractive}, we omit the details here.

\subsection{Application: communication complexity and streaming algorithms}

\subsubsection{Communication complexity: Hidden Matching and its variants}


The Boolean Hidden Matching ($\BHM$) problem was introduced  by Bar-Yossef et al.~\cite{bar2004exponential} (which was in turn inspired by Kerenidis and de Wolf~\cite{DBLP:journals/jcss/KerenidisW04} for proving $\mathsf{LDC}$ lower bounds) in order to prove exponential separations between quantum and classical one-way communication complexities. {We describe below} the generalized Hidden Matching problem over larger alphabets and hypermatching. {Here a hypermatching is a set of hyperedges without common vertices}. 
The $r$-ary Hidden Hypermatching ($r\text{-}\HH(\alpha,t,n)$) problem is a two-party communication problem between Alice and Bob: Alice is given $x\in \Z_r^n$ and Bob is given a string $w\in \Z_r^{\alpha n/t}$ and $\alpha n/t$-many disjoint $t$-tuples from $[n]$ (where $\alpha\in(0,1]$). {The disjoint $t$-tuples from $[n]$ can be viewed as hyperedges of size $t$, which, in turn,} can also be viewed as an {incidence} matrix $M\in \{0,1\}^{\alpha n/t\times n}$ (each row corresponding to a hyperedge). In the $\YES$ instance it is promised that $w=Mx$ (over $\Z_r$), while in the $\NO$ instance it is promised that $w$ is uniformly random, and the goal is to decide which is the case using a message sent from Alice to Bob.

There have been a few lines of work {on} understanding the problem of Hidden Hypermatching: $(i)$ the seminal work of Bar-Yossef et al.~\cite{bar2004exponential} and Gavinsky et al.~\cite{gavinsky2007exponential} showed that, for $r=t=2$, $\BHM$ can be solved using $O(\log n)$ qubits but requires $\Omega(\sqrt{n})$ classical bits of communication; $(ii)$ Verbin and Yu~\cite{verbin2011streaming} considered the problem where $r=2$ and $t\geq 2$ (which in fact inspired many follow-up works on using hypermatching for classical streaming lower bounds) and showed a classical lowed bound of $\Omega(n^{1-1/t})$, which was subsequently generalized to a $\Omega(n^{1-2/t})$ quantum lower bound by Shi, Wu, and Yu~\cite{shi2012limits}; $(iii)$ Guruswami and Tao~\cite{DBLP:conf/approx/GuruswamiT19} studied the problem for when $t=2$ and $r\geq 2$, proving a classical $\Omega(\sqrt{n})$ lower bound; $(iv)$ apart from these, there have been more exotic generalizations of the Hidden Hypermatching problem, e.g., Kapralov, Khanna, and Sudan~\cite{kapralov2014streaming} proposed the Boolean Hidden Partition, where Bob does not receive a matching anymore {(a graph made up of pair-wise disjoint edges)}, but the edges of any graph $G$, and Doriguello and Montanaro~\cite{doriguello2020exponential} expanded the $2\text{-}\HH(\alpha,t,n)$ problem to computing a fixed Boolean function on the hyperedges of Bob's hypermatching instead of the Parity function. Here we shall consider the standard Hidden Hypermatching problem over a larger alphabet and give both upper and lower bounds for its quantum and classical communication complexities for every $t,r\geq 2$.

\paragraph{Upper bounds on Hidden Hypermatching.} For a given $t\geq2$, the same classical communication protocol for $r=2$ can be used for general $r>2$. The idea is that Alice picks $O((n/\alpha)^{1-1/t})$ entries of $x$ uniformly at random to send to Bob. By the Birthday Paradox, with high probability Bob will obtain all the values from one of his hyperedges $i$, and thus can compare $(Mx)_i$ with the corresponding $w_i$. If they are equal, he outputs $\YES$, otherwise he outputs $\NO$, which leads to an one-side error of $O(1/r)$. The total amount of communication is $O(\log{(rn)}(n/\alpha)^{1-1/t})$ bits.\footnote{ One can  further improve this complexity to $O(\log{(n\log{r})} + (\log r)\cdot (n/\alpha)^{1-1/t})$ by Newman's theorem~\cite{newman1991private}.} The situation is more interesting in the quantum setting. For $t=2$, we prove that Hidden Hypermatching can be solved using only a logarithmic amount of qubits for every $r=\poly(n)$.
\begin{result}[\cref{thr:rary-upper}]
\label{res:upperBHM}
    There is a one-way protocol for $r\text{-}\HH(\alpha,2,n)$ with one-sided error $1/3$ using $O(\log{(nr)}/\alpha)$~qubits.
\end{result}
The above bound uses a non-trivial procedure that allows {one} to learn the sum of two numbers modulo $r$ by using just one ``query'' and crucially uses the knowledge of the string $w$: given a suitable superposition of two numbers, one can obtain their sum with one-sided error by using one measurement. As far as we are aware, such a statement was not known prior to our work. However, the knowledge of $w$ is vital, which means that the protocol does not work for more general settings where there is no promise on the inputs (e.g., a relational version of the $r$-ary Hidden Hypermatching problem where Bob must output one hyperedge $i$ and its corresponding value $(Mx)_i$), and it also cannot be used as a building block for the general case $t,r>2$. The current upper bound on the quantum communication complexity of the $r\text{-}\HH(\alpha,t,n)$ problem with $t,r> 2$ thus matches the classical one. In view of the lower bounds stated below {and the known upper bound $O(n^{1-1/\lceil t/2\rceil}\log{n})$ for $r=2$~\cite{DBLP:journals/jcss/KerenidisW04} (note that $\lceil t/2\rceil$ is the number of quantum queries required to learn the parity of $t$ bits~\cite{farhi1998limit,Beals2001quantum})}, we make the following conjecture. 
\begin{conjecture}
    If $t,r>2$, there is a protocol for $r\text{-}\HH(\alpha,t,n)$ using $O(\log{(rn)}(n/\alpha)^{1-1/\lceil t/2\rceil})$~qubits.
\end{conjecture}

\paragraph{Lower bounds on Hidden Hypermatching.} The standard approach for proving a lower bound on the amount of communication required to solve the Hidden Hypermatching problem is via Fourier analysis. In the classical proofs of Gavinsky et al.~\cite{gavinsky2007exponential}, Verbin and Yu~\cite{verbin2011streaming}, and Guruswami and Tao~\cite{DBLP:conf/approx/GuruswamiT19}, the total variation distance between the probability distributions arising from the $\YES$ and $\NO$ instances is bounded using the inequality of Kahn, Kalai, and Linial~\cite{kahn1989influence} (which can be seen as a corollary of the hypercontractivity inequality). On the other hand, Shi, Wu, and Yu~\cite{shi2012limits} obtained a quantum lower bound by bounding the Schatten 1-norm between the possible density matrices received by Bob in both $\YES$ and $\NO$ instances via the matrix-valued hypercontractivity from Ben-Aroya, Regev, and de Wolf~\cite{ben2008hypercontractive}. We follow a similar approach by using our \emph{generalized matrix-valued hypercontractive} inequality from \cref{result:hyper} in order to obtain the following lower bound (for $r=2$ our lower bound is exponentially better in $\alpha$ than~\cite{shi2012limits}).
\begin{result}[\cref{thm:hypermatchinglowerbound}, \cref{thr:classical_communication_lower_bound}]
    \label{res:lowerHHM}
    Every constant-bias protocol for $r\text{-}\HH(\alpha,t,n)$ with $t,r\geq 2$ requires $\Omega((n/t)^{1-2/t}/\alpha^{2/t})$ qubits of communication or $\Omega((n/t)^{1-1/t}/\alpha^{1/t})$ bits of communication.
\end{result}
The phase transition from $t=2$ to $t>2$ in the quantum communication complexity of $r$-$\HH(\alpha,t,n)$ is reminiscent of the well-known fact that parity of $n$ bits can be computed exactly using $\lceil n/2\rceil$ quantum queries~\cite{farhi1998limit,Beals2001quantum}: a function requiring more than one query within the same hyperedge will lead to a large communication complexity since, by design of the Hidden Matching problem, it is hard to extract information from the same hyperedge more than once.

\subsubsection{Streaming lower bounds}
One-way communication complexity lower bounds has been used by several recent works to prove streaming lower bounds~\cite{verbin2011streaming,kapralov2014streaming,guruswami2017streaming,DBLP:conf/approx/GuruswamiT19,DBLP:conf/focs/ChouGV20}. To see this, suppose a problem $P$ has inputs $(X,Y)$ and the goal is to find space-efficient streaming algorithms to compute $P(X,Y)$, when $X,Y$ are presented in a stream (i.e., presented bit-by-bit). Then, one way to lower bound the \emph{space complexity} is to prove lower bounds on the following problem: consider the one-way communication problem where Alice gets the input $X$, Bob gets the input $Y$, their goal is to compute $P(X,Y)$ and only Alice is allowed to communicate to Bob. Then one can show that any lower bound for randomized one-way communication implies an equivalent lower bound for streaming algorithms. This technique has been used by a sequence of papers to prove lower bounds on space complexity of Max-Cut~\cite{kapralov2014streaming,kapralov2019optimal,DBLP:conf/approx/GuruswamiT19}, matching~\cite{goel2012communication}, Max-CSP~\cite{guruswami2017streaming,DBLP:conf/focs/ChouGV20,chou2021approximabilityB} and counting cycles~\cite{verbin2011streaming,assadi2021simple}. One problem that is used often in this direction is the aforementioned Boolean Hidden Matching and related variants. Using our classical and quantum communication lower bounds (\cref{res:lowerHHM}) we present two lower bounds for streaming problems on generalizations to Max-Cut.

There are a few natural generalizations to Max-Cut. One is Max-$t$-Cut, i.e., finding the maximum cut value on a hypergraph with $t$-sized hyperedges, to which the $\Omega(n)$ classical lower bound from~\cite{kapralov2019optimal} already applies. Another is the Unique Games problem, a constraint satisfaction problem defined on a graph, where a linear constraint (a permutation) over $\Z_r$ is specified on each edge and the goal is to find a vertex assignment over $\Z_r$ that maximizes the number of satisfied constraints. When $r=2$, Unique Games reduces to Max-Cut. Guruswami and Tao~\cite{DBLP:conf/approx/GuruswamiT19} studied the streaming complexity of the Unique Games problem and proved a lower bound of $\Omega(\sqrt{n})$ in the adversarial model by using a reduction {from} Hidden Matching over $\Z_r$. This bound was improved to $\Omega(n)$ by Chou et al.~\cite{chou2022linear} for a larger set of problems including Unique Games.

Here we join both directions, i.e., Max-$t$-Cut and the standard Unique Games problem, into a generalized version of Unique Games defined on a hypergraph and obtain a streaming \emph{quantum} lower bound in the adversarial model for any value $t,r\geq 2$.\footnote{We note that Kapralov et al.~\cite{kapralov2014streaming} proved a classical lower bound $\Omega(n^{1-\varepsilon})$ for $(1+\varepsilon)$-approximations of Max-Cut in the adversarial model and their proof, together with \cref{res:lowerHHM}, readily implies a similar quantum lower bound.} 
\begin{result}[\cref{thm:streaminguglowerbound}]
    Every streaming algorithm giving a $(r-\varepsilon)$-approximation for Unique Games on $t$-hyperedge $n$-vertex hypergraphs over $\Z_r$ uses $\Omega(n^{1-2/t})$ quantum space.
\end{result}
The above result follows from generalizing the proof of Guruswami and Tao~\cite{DBLP:conf/approx/GuruswamiT19}. 
As far as we are aware, these are the first quantum lower bounds for Unique Games in the streaming model.

\subsection{Locally decodable codes}

A locally decodable  code ($\mathsf{LDC}$) is an error correcting code that allows to retrieve a single bit of the original message (with high probability) by only examining a few bits in a corrupted codeword. More formally, a $(q,\varepsilon,\delta)$-$\mathsf{LDC}$ was defined by Katz and Trevisan~\cite{katz2000efficiency} as a function $C:\Z_r^n\rightarrow \Z_r^N$ that satisfies the following: for all $x\in \Z_r^n$, $i\in [n]$ and $y\in\Z_r^N$ that satisfies $d(C(x),y)\leq \delta$ (i.e., a $\delta$-fraction of the elements of $C(x)$ are corrupted), there exists an algorithm $\A$ that makes $q$ queries to $y$ non-adaptively and outputs $\mathcal{A}^{y}(i)\in \Z_r$ such that $\Pr[\mathcal{A}^{y}(i)=x_i]\geq 1/r+\varepsilon$ (where the probability is over the randomness of $\A$). Over $\{0,1\}$, $\mathsf{LDC}$s have found several applications in private information retrieval~\cite{chor1995private}, multiparty computation~\cite{ishai2004hardness}, data structures~\cite{chen2009efficient}, and average-case complexity theory~\cite{trevisan2004some}.

The natural question in constructing $\mathsf{LDC}$s is the trade-off between $N$ and $n$. A well-known $2$-query $\mathsf{LDC}$ is the Hadamard encoding that maps $x\in \Z_r^n$ into the string $C(x)=(\langle x,y\rangle)_{y\in \{0,1\}^n}$: on input $i\in [n]$, a decoding algorithm queries $C(x)$ at a uniformly random $y$ and $y\oplus e_i$ and retrieves $C(x)=\langle x,y\oplus e_i\rangle-\langle x,y\rangle$, where $e_i = 0^{i-1}10^{n-i}$. Here the encoding length is $N=2^n$, and an important question is, are there $2$-query $\mathsf{LDC}$s with $N\ll 2^n$? For the case $r=2$, Goldreich et al.~\cite{goldreich2002lower} showed a lower bound $N= 2^{\Omega(n)}$ for \emph{linear codes}, which was later improved by Obata~\cite{obata2002optimal}. Later, Kerenidis and de Wolf~\cite{DBLP:journals/jcss/KerenidisW04} proved an exponential lower bound for \emph{non-linear codes} using a quantum argument!\footnote{For simplicity in exposition, we omit the dependence on $\delta,\varepsilon$ in these lower bounds.} 

This left open the setting where $r>2$. Following these works, {for $2$-query \emph{non-linear} $\mathsf{LDC}$s $C:\{0,1\}^n\to\Z_r^N$ (note the inputs are over $\{0,1\}$ and not $\Z_r$), Wehner and de Wolf~\cite{wehner2005improved} proved the lower bound $N=2^{\Omega(n/r^2)}$.} On the other hand, Dvir and Shpilka~\cite{dvir2007locally} showed a lower bound of $N= 2^{\Omega(n)}$ for every $2$-query \emph{linear} $\mathsf{LDC}$ $C:\Z_r^n\to\Z_r^N$, even independent of the field size. To prove their result, they crucially observed that, given a linear $\mathsf{LDC}$ over $\Z_r$, one can construct a linear $\mathsf{LDC}$ over $\{0,1\}$ (with almost the same parameters) and then invoked the result of Goldreich et al.~\cite{goldreich2002lower}. This reduction, however, fails for non-linear codes and {motivates the question of whether} there are
\emph{non-linear} $\mathsf{LDC}$s $C:\Z_r^n\to\Z_r^N$ with $N\ll 2^n$.

The main contribution here is a lower bound for \emph{non-linear} $\mathsf{LDC}$s over $\Z_r$ that scale as $2^{\Omega(n/r^2)}$, and which gives a super-polynomial lower bound for $r=o(n^{1/2})$. Our lower bound comes from using the non-commutative Khintchine's inequality~\cite{tropp2015introduction}.\footnote{We thank Jop Bri\"{et} for helping us prove this theorem.}

%
\begin{result}[\cref{thm:2querylowerboundhyper}]
\label{res:ldchyper}
	If $C:\Z_r^n\to\Z_r^N$ is a $(2,\delta,\varepsilon)$-$\mathsf{LDC}$, then $N=2^{\Omega(\delta^2\varepsilon^2 n/r^2)}$.
\end{result}
We briefly mention that, if one requires the success probability to be larger than, for example, $1/2+\varepsilon$ instead of $1/r+\varepsilon$, so that plurality vote can be used and the success probability amplified, then $\varepsilon$ becomes a constant bounded away from $1/r$ (if $r>2$) and our lower bound is no longer dependent on $\varepsilon$. {Moreover, we stress that the success probability is bounded by $1/r + \varepsilon$, meaning that the lower bound of Wehner and de Wolf~\cite{wehner2005improved} is not directly applicable by restricting the messages to binary strings}. Finally, we remark that in an earlier version of the paper, we proved a similar lower bound using our matrix-valued hypercontractivity via the proof technique in~{\rm \cite[Theorem~11]{ben2008hypercontractive}}. This leads, however, to a weaker bound, $N= 2^{\Omega(\delta^2 \varepsilon^4 n/r^4)}$. After our preprint appeared online, we were made aware of the non-commutative Khintchine's inequality which leads us to the improved lower bound of $2^{\Omega(\delta^2\varepsilon^2n/r^2)}$.

\paragraph{Further applications (Private information retrieval).} Katz and Trevisan~\cite{katz2000efficiency} and Goldreich et al.~\cite{goldreich2002lower} established a nice connection between $\mathsf{LDC}$s and private information retrieval (\textsf{PIR}) protocols. We do not define these $\textsf{PIR}$ schemes here and refer the reader to \cref{sec:PIR}. Using \cref{res:ldchyper} and other auxiliary results, we get the following lower bound for $\textsf{PIR}$ schemes over $\Z_r$ when $r\geq 2$ is prime. 
\begin{result}[\cref{thr:pir_lower_bound}]\label{res:pir}
    Let $r\geq 2$ be prime. A classical $2$-server $\mathsf{PIR}$ scheme with alphabet $\Z_r$, query size $t$, answer size $a$, and recovery probability $1/r + \varepsilon$ satisfies $t= \Omega\big(\varepsilon^2 n/r^{4a+6} - a)$.
\end{result}

{Previous lower bounds on $2$-server $\mathsf{PIR}$ protocols with alphabet size $r=2$ include, among others, $t = \Omega(n\varepsilon^2/2^{5a})$~\cite{DBLP:journals/jcss/KerenidisW04}, $t = \Omega(n\varepsilon^2/2^{2a})$~\cite{wehner2005improved}, and (for a restricted model of 2-server $\mathsf{PIR}$) $t+a = \Omega(n^{1/3})$~\cite{Razborov2007lower}. Other bounds in various settings include~\cite{goldreich2006lower,beigel2006tight,Yeo2023lower}. We note that \cref{res:pir} does not contradict Dvir and Gopi's~\cite{dvir20162} 2-server PIR scheme with query size $t = n^{O(\sqrt{\log\log{n}/\log{n}})} = n^{o(1)}$ since it uses answer size $a = O(t)$, in which case our bound becomes trivial.}

\paragraph{After completion of this work.}
Kallaugher and Parekh~\cite{kallaugher2022quantum} recently put out an online preprint in which they improve our quantum streaming lower bound for Max-$t$-Cut to $\Omega(n)$. As far as we know, our quantum streaming lower bound for $\Z_r$ when $r>2$ was not known before.

\subsection{Future work}

\emph{\textbf{1.}} \emph{Proving $\mathsf{LDC}$ lower bounds.} The first natural open question is, can we prove a lower bound of $N= 2^{\Omega(n/r)}$ for $\mathsf{LDC}$s over $\Z_r$, or, more ambitiously, prove that $N= 2^{\Omega(n)}$? As far as we are aware, there are no super-polynomial lower bounds for $N$ even for $r=\omega(\sqrt{n})$. 
Similarly, can one also prove a lower bound of $N= 2^{\Omega(n\log r)}$ for  \emph{non-linear} locally-\emph{correctable} codes over $\Z_r$ (thereby matching a similar lower bound for linear case~\cite{bhattacharyya2011tight})?

\emph{\textbf{2.}} \emph{Communication complexity of $r$-ary Hidden Hypermatching.} Our communication protocol behind \cref{res:upperBHM} relies on the promise on the inputs, i.e., on the string $w\in\mathbb{Z}_r^{\alpha n/t}$ that either satisfies $w=Mx$ or is uniformly random. Is there a protocol with the same complexity which does not explicitly use $w$? More generally, what is the communication complexity of a relational version of the $r$-$\HH(\alpha,2,n)$ problem in which Bob outputs a hyperedge and the corresponding entry of $Mx$? Moreover, is it possible to match the quantum lower bounds from \cref{res:lowerHHM}?

\emph{\textbf{3.}} \emph{Better bounds on streaming algorithms.} What is the quantum space complexity of approximating Unique Games? {Given the recent work of Kallaugher and Parekh~\cite{kallaugher2022quantum}, who proved a $\Omega(n)$ streaming lower bound for Max-$t$-Cut, we conjecture that the same bound applies to Unique Games.} A proof might require obtaining new quantum lower bounds for the communication problems introduced in~\cite{kapralov20171,kapralov2019optimal,chou2021approximabilityB,chou2022linear}.

\emph{\textbf{4.}} \emph{Generalized hypercontractivity.} Another open question is regarding our main \cref{result:hyper}, which shows a form of $(2,p)$-hypercontractivity, since the result works for all Schatten $p$-norms with $p\in [1,2]$. Can we prove a general $(q,p)$-hypercontractive statement for matrices, firstly for matrix-valued functions over $\{0,1\}$, and then further generalize that to functions over $\Z_r$? Proving this might also require a different generalization of the inequality of Ball, Carlen, and Lieb~\cite{ball1994sharp}.

\section{Preliminaries}
\label{sec:sec2}

Let $[n]\triangleq\{1,\ldots,n\}$. For $r\in\Z$, $r\geq 2$, we let $\Z_r\triangleq\{0,\ldots,r-1\}$ be the ring with addition and multiplication modulo $r$, and let $\omega_r \triangleq e^{2\pi i/r}$. 
Given $S\in\Z_r^n$, we write $|S|\triangleq |\{i\in[n]: S_i\neq 0\}|$ for its Hamming weight. 
Let $\operatorname{D}(\mathbb{C}^m)$ be the set of all quantum states over $\mathbb{C}^m$, i.e., the set of positive semi-definite matrices with trace $1$. For a matrix $M\in \mathbb{C}^{m\times m}$, its (unnormalized) Schatten $p$-norm is defined as $\|M\|_p \triangleq (\operatorname{Tr}|M|^p)^{1/p} = \big(\sum_{i=1}^m\sigma_i(M)^p\big)^{1/p}$, where $\{\sigma_i(M)\}_i$ are the singular values of $M$, i.e., the eigenvalues of the positive semi-definite operator $|M|\triangleq\sqrt{M^\dagger M}$. Moreover, let $\|M\| \triangleq \max_{i\in[m]}\sigma_i(M)$ and $\|M\|_F \triangleq \sqrt{\sum_{i,j=1}^m |M_{ij}|^2}$ be its spectral and Frobenious norms, respectively (notice that $\|M\| = \|M\|_\infty$). 
Given a vector $v\in\mathbb{C}^m$, its $p$-norm is $\|v\|_p \triangleq \big(\sum_{i=1}^m |v_i|^p\big)^{1/p}$. Given two probability distributions $P$ and $Q$ on the same finite set, their total variation distance is $\|P-Q\|_{\text{tvd}} \triangleq \sum_i |P(i) - Q(i)|$ (we might abuse notation and use random variables instead of their probability distributions in $\|\cdot\|_{\text{tvd}}$). For a probability $p = 1/r + \varepsilon$ with fixed $r\in\Z$, we refer to $\varepsilon$ as its \emph{advantage}, and to {$\varepsilon r/(r-1)$} as its \emph{bias}. 

The Fourier transform of a matrix-valued function $f:\Z_r^n\to\mathbb{C}^{m\times m}$ is a function $\widehat{f}:\Z_r^n\to\mathbb{C}^{m\times m}$ defined by
\begin{align*}
    \widehat{f}(S) = \frac{1}{r^n}\sum_{x\in\Z_r^n}f(x)\omega_r^{-S\cdot x},
\end{align*}
where $S\cdot x = \sum_{i=1}^n S_ix_i$  is a sum over $\Z_r$. Here the Fourier coefficients $\widehat{f}(S)$ are $m\times m$ complex matrices and we can write $f:\Z_r^n\to\mathbb{C}^{m\times m}$ as
\begin{align*}
    f(x) = \sum_{S\in\Z_r^n}\widehat{f}(S)\omega_r^{S\cdot x}.
\end{align*}

We shall need the Holevo-Helstrom theorem~\cite{helstrom1976quantum} which characterizes the optimal success probability of distinguishing between two quantum states.
\begin{fact}[{\cite[Theorem~3.4]{watrous2018theory}}]
    \label{lem:lem3.5.c3}
    Let $\rho_0,\rho_1$ be two quantum states that appear with probability $p$ and $1-p$, respectively. The optimal success probability of predicting which state it is by a POVM~is
    \begin{align*}
        \frac{1}{2} + \frac{1}{2}\|p\rho_0 - (1-p)\rho_1\|_1.
    \end{align*}
\end{fact}

We will need the following result derived from the non-commutative Khintchine's inequality.
\begin{lemma}
    \label{lem:lem1}
    Given $A_1,\dots,A_n\in\mathbb{C}^{N\times N}$, then
    \begin{align*}
        \operatorname*{\mathbb{E}}_{k\sim\mathbb{Z}_r^n}\left\|\sum_{i=1}^n \omega_r^{k_i}A_i\right\| \leq 2\sqrt{2\log(2N)}\sqrt{\sum_{i=1}^n \|A_i\|^2}.
    \end{align*}
\end{lemma}
\begin{proof}
    $\begin{aligned}[t]
        \operatorname*{\mathbb{E}}_{k\sim\mathbb{Z}_r^n}\left\|\sum_{i=1}^n \omega_r^{k_i}A_i\right\| &= \operatorname*{\mathbb{E}}_{k\sim\mathbb{Z}_r^n}\left\|\sum_{i=1}^n \left(\omega_r^{k_i}A_i -\operatorname*{\mathbb{E}}_{k'\sim\mathbb{Z}_r^n}\big[\omega_r^{k'_i}A_i\big]\right) \right\|\\
        &\leq \operatorname*{\mathbb{E}}_{k,k'\sim\mathbb{Z}_r^n}\left\|\sum_{i=1}^n (\omega_r^{k_i} - \omega_r^{k'_i})A_i \right\|\\
        &= \operatorname*{\mathbb{E}}_{k,k'\sim\mathbb{Z}_r^n}\operatorname*{\mathbb{E}}_{\varepsilon\sim\{\pm 1\}^n}\left\|\sum_{i=1}^n \varepsilon_i(\omega_r^{k_i} - \omega_r^{k'_i})A_i \right\|\\
        &\leq 2\operatorname*{\mathbb{E}}_{k\sim\mathbb{Z}_r^n}\operatorname*{\mathbb{E}}_{\varepsilon\sim\{\pm 1\}^n}\left\|\sum_{i=1}^n \varepsilon_i\omega_r^{k_i}A_i \right\|\\
        &\leq 2\sqrt{2\log(2N)}\sqrt{\sum_{i=1}^n \|A_i\|^2},
    \end{aligned}$

    \noindent using the non-commutative Khintchine's inequality~\cite[Theorem~4.1.1]{tropp2015introduction} in the last step.
\end{proof}

\section{Hypercontractive Inequality}
In this section we prove our main result, a hypercontractive inequality for matrix-valued functions over $\Z_r$, generalizing a result from~\cite{ben2008hypercontractive}. The proof is by induction on $n$ and the base case $n=1$ is proven in \cref{sec:sec3.1}, which is a generalization of Ball, Carlen, and Lieb~\cite{ball1994sharp} when considering $r$ matrices. After this, the induction is fairly straightforward and is described in \cref{sec:sec3.2}.

\subsection{Generalizing Ball, Carlen, and Lieb}
\label{sec:sec3.1}
We first state the powerful inequality of Ball, Carlen, and Lieb~\cite[Theorem~1]{ball1994sharp}.
\begin{theorem}[Optimal 2-uniform convexity]
    \label{thr:thr1}
    Let $A,B\in\mathbb{C}^{n\times n}$, and $p\in[1,2]$. Then
    $$
        \left(\frac{\|A+B\|_p^p + \|A-B\|_p^p}{2}\right)^{2/p} \geq \|A\|_p^2 + (p-1)\|B\|_p^2.
    $$
\end{theorem}
 As previously mentioned in the introduction, this inequality was first proven by Tomczak-Jaegermann~\cite{tomczak1974moduli} for $p\leq 4/3$, before being generalized by Ball, Carlen, and Lieb~\cite{ball1994sharp} for all $p\in [1,2]$ in 1994. Since then it has found several applications~\cite{carlen1993optimal,duchi2010composite,lee2004embedding,naor2016spectral}. The above result can be recast in a slightly different way.
\begin{theorem}
    \label{thr:alternativeBCL}
    Let $p\in[1,2]$ and $Z,W\in\mathbb{C}^{n\times n}$ such that $\operatorname{Tr}[|Z|^{p-2}ZW^\dagger] =\operatorname{Tr}[|Z|^{p-2}WZ^\dagger] = 0$ (where $|Z|^{p-2}=(ZZ^\dagger)^{p/2-1}$).~Then
    \begin{align*}
        \|Z+W\|_p^2 \geq \|Z\|_p^2 + (p-1)\|W\|_p^2.
    \end{align*}
\end{theorem}
\cref{thr:alternativeBCL} is implicit in the proof of~\cite[Theorem~1]{ball1994sharp}, and it is where most of the difficulty lies, while the reduction from \cref{thr:thr1} to \cref{thr:alternativeBCL} is done by defining
\begin{align*}
    Z = \begin{bmatrix} A & 0 \\ 0 & A \end{bmatrix}, \qquad W = \begin{bmatrix} B & 0 \\ 0 & -B \end{bmatrix}.
\end{align*}
Nonetheless, \cref{thr:alternativeBCL} holds more generally for any $Z,W\in\mathbb{C}^{n\times n}$ that satisfy $\operatorname{Tr}[|Z|^{p-2}ZW^\dagger] =\operatorname{Tr}[|Z|^{p-2}WZ^\dagger]=0$. By using this result, we can prove the following generalization of \cref{thr:thr1}.
\begin{theorem}[A generalization of~\cite{ball1994sharp}]
\label{eq:conjectureballgen}
    Let $r\in\mathbb{Z}$, $r\geq 2$. Let $\omega_r \triangleq e^{2i\pi/r}$, $A_0,\ldots,A_{r-1}\in \mathbb{C}^{n\times n}$, and $p\in [1,2]$, then
    \begin{subequations}
	\begin{align}
		\left(\frac{1}{r} \sum_{j=0}^{r-1}\|A_j\|_p^p\right)^{2/p} &\geq \left\| \frac{1}{r}\sum_{j=0}^{r-1} A_j \right\|_p^2 + \frac{(p-1)(1-(p-1)^{r-1})}{(r-1)(2-p)}\sum_{k=1}^{r-1} \left\| \frac{1}{r}\sum_{j=0}^{r-1} \omega_r^{-jk}A_j \right\|_p^2,\\
		\left(\frac{1}{r} \sum_{k=0}^{r-1}\left\|\sum_{j=0}^{r-1} \omega_r^{jk} A_j\right\|_p^p\right)^{2/p} &\geq\left\|A_0 \right\|_p^2 + \frac{(p-1)(1-(p-1)^{r-1})}{(r-1)(2-p)}\sum_{k=1}^{r-1} \left\| A_k\right\|_p^2.\label{eq:genball}
	\end{align}
	\end{subequations}
\end{theorem}
Notice that for $r=2$ we recover \cref{thr:thr1}, since $\frac{(p-1)(1-(p-1)^{r-1})}{(r-1)(2-p)} = p-1$.
\begin{proof}
    In order to prove this theorem, first note that both inequalities are equivalent: just define $A'_k = \frac{1}{r}\sum_{j=0}^{r-1} \omega_r^{-jk}A_j \iff A_k = \sum_{j=0}^{r-1}\omega_r^{jk}A'_j$. Therefore we shall focus on \cref{eq:genball}.	In order to prove it, let us first define the $rn\times rn$ matrices
    \begin{align}
        \label{eq:matricesY}
        Z_j \triangleq \operatorname{diag}(\{\omega_r^{jk}A_j\}_{k=0}^{r-1}) = \begin{bmatrix}
            A_j & 0 & 0 & \dots & 0\\
            0 & \omega_r^{j}A_j & 0 & \dots & 0\\
            0 & 0 & \omega_r^{2j}A_j & \dots & 0\\
            \vdots & \vdots & \vdots & \ddots & \vdots\\
            0 & 0 & 0 & \dots & \omega_r^{(r-1)j}A_{j}
        \end{bmatrix}
    \end{align}
    for $j\in\{0,\dots,r-1\}$. Now, since the trace is additive for block matrices, we have
    \begin{align}
        \label{eq:YandAmatrices}
        \operatorname{Tr}\left|\sum_{j=0}^{r-1}Z_j\right|^p =  \sum_{k=0}^{r-1}\operatorname{Tr}\left|\sum_{j=0}^{r-1} \omega_r^{jk} A_j\right|^p.
    \end{align}
    Moreover, observe that
    \begin{align*}
        \|Z_j\|_p^2 = \left(\sum_{k=0}^{r-1}\operatorname{Tr}|\omega_r^{jk}A_j|^p\right)^{2/p} = (r\operatorname{Tr}|A_j|^p)^{2/p} = r^{2/p}\|A_j\|_p^2.
    \end{align*}
    Therefore we can rewrite \cref{eq:genball} as
    \begin{align*}
        \left\|\sum_{j=0}^{r-1}Z_j\right\|^2_p \geq \|Z_0\|_p^2 + \frac{(p-1)(1-(p-1)^{r-1})}{(r-1)(2-p)}\sum_{j=1}^{r-1}\|Z_j\|_p^2.
    \end{align*}
    The above can be proven by repeated applications of \cref{thr:alternativeBCL} as follows: consider a permutation of $[r-1]$ given by $(k_1,\ldots,k_{r-1})$. Since 
    \begin{align*}
        \operatorname{Tr}[|Z_j|^{p-2}Z_jZ_k^\dagger] &= \operatorname{Tr}\begin{bmatrix}
            |A_j|^{p-2}A_jA_k^\dagger & 0 & \dots & 0\\
            0 & \omega_r^{j-k}|A_j|^{p-2}A_jA_k^\dagger & \dots & 0\\
            \vdots & \vdots & \ddots & \vdots\\
            0 & 0 & \dots & \omega_r^{(r-1)(j-k)}|A_j|^{p-2}A_jA_k^\dagger
        \end{bmatrix} \\
        &= \operatorname{Tr}[|A_j|^{p-2}A_jA_k^\dagger] \sum_{\ell = 0}^{r-1}\omega_r^{\ell(j-k)} = 0
    \end{align*}
    for any $j\neq k$ (and similarly $\operatorname{Tr}[|Z_j|^{p-2}Z_kZ_j^\dagger] = 0$), then (define $k_0\triangleq0$)
	\begin{align*}
	    \operatorname{Tr}\left[|Z_{k_j}|^{p-2}Z_{k_j}\left(\sum_{l=j+1}^{r-1}Z_{k_l}\right)^\dagger\right] = \operatorname{Tr}\left[|Z_{k_j}|^{p-2}\left(\sum_{l=j+1}^{r-1}Z_{k_l}\right)Z_{k_j}^\dagger\right] = 0
	\end{align*}
	for every $j\in\{0,1,\dots,r-2\}$, meaning that \cref{thr:alternativeBCL} can be applied, which implies
	\begin{align*}
	    \left\|\sum_{j=0}^{r-1}Z_j\right\|^2_p &\geq \|Z_0\|_p^2 + (p-1)\left\| \sum_{j=1}^{r-1}Z_j\right\|^2_p\\
	    &\geq \|Z_0\|_p^2 + (p-1)\|Z_{k_1}\|_p^2 + (p-1)^2\left\| \sum_{j=2}^{r-1}Z_{k_j}\right\|^2_p\geq \|Z_0\|_p^2 + \sum_{j=1}^{r-1}(p-1)^j\|Z_{k_j}\|_p^2.
	\end{align*}
	Averaging the above inequality over all the $(r-1)!$ permutations of the set $[r-1]$, we obtain
	\begin{align*}
	    \left\|\sum_{j=0}^{r-1}Z_j\right\|^2_p &\geq \|Z_0\|_p^2 + \frac{1}{(r-1)!}\sum_{j=1}^{r-1}\|Z_j\|_p^2\sum_{k=1}^{r-1}(r-2)!(p-1)^k\\
	    &= \|Z_0\|_p^2 + \frac{(p-1)(1-(p-1)^{r-1})}{(r-1)(2-p)}\sum_{j=1}^{r-1}\|Z_j\|_p^2,
	\end{align*}
proving our theorem statement.
\end{proof}   
\begin{remark}
    \label{rem:1}
    It is not hard to see that $\frac{(p-1)(1-(p-1)^{r-1})}{(r-1)(2-p)} \geq \frac{p-1}{r-1}$ and $\lim_{p\to 2}\frac{(p-1)(1-(p-1)^{r-1})}{(r-1)(2-p)} = 1$.
\end{remark}
\noindent Observe that $t\mapsto t^{p/2}$ is concave for $p\in[1,2]$, hence \cref{eq:conjectureballgen} implies the seemingly weaker
\begin{align}
    \label{eq:BCLweaker}
    \frac{1}{r} \sum_{k=0}^{r-1}\left\|\sum_{j=0}^{r-1} \omega_r^{jk} A_j\right\|_p^2 &\geq \left\|A_0 \right\|_p^2 + \frac{(p-1)(1-(p-1)^{r-1})}{(r-1)(2-p)}\sum_{k=1}^{r-1} \left\| A_k\right\|_p^2
\end{align}
for $p\in[1,2]$. Nonetheless, the above inequality also implies \cref{eq:conjectureballgen} (this fact was already pointed out for $r=2$ by~\cite{ball1994sharp}). Indeed, consider again the $rn\times rn$ matrices $Z_j$ from \cref{eq:matricesY}. Then, similar to \cref{eq:YandAmatrices} (which only considered the $\ell=0$ case below), for any $\ell\in\Z_r$ we have
\begin{align*}
    \operatorname{Tr}\left|\sum_{j=0}^{r-1}\omega_r^{j\ell} Z_j\right|^p = \sum_{k=0}^{r-1}\operatorname{Tr}\left|\sum_{j=0}^{r-1} \omega_r^{jk} A_j\right|^p \implies \left\|\sum_{j=0}^{r-1}\omega_r^{j\ell} Z_j\right\|_p^2 = \left(\sum_{k=0}^{r-1}\left\|\sum_{j=0}^{r-1} \omega_r^{jk} A_j\right\|^p_p\right)^{2/p}.
\end{align*}
 Since $\|Z_j\|^2_p = r^{2/p}\|A_j\|^2_p$ for $j\in\Z_r$, \cref{eq:BCLweaker} implies (define $\zeta \triangleq \frac{(p-1)(1-(p-1)^{r-1})}{(r-1)(2-p)}$ for simplicity)
\begin{align*}
    \left\|A_0 \right\|_p^2 + \zeta\sum_{k=1}^{r-1} \left\| A_k\right\|_p^2 = \frac{\left\|Z_0 \right\|_p^2}{r^{2/p}} + \zeta\sum_{k=1}^{r-1} \frac{\left\| Z_k\right\|_p^2}{r^{2/p}}
    \leq \frac{r^{-2/p}}{r} \sum_{\ell=0}^{r-1}\left\|\sum_{j=0}^{r-1} \omega_r^{j\ell} Z_j\right\|_p^2
    = \left(\frac{1}{r}\sum_{k=0}^{r-1}\left\|\sum_{j=0}^{r-1} \omega_r^{jk} A_j\right\|^p_p\right)^{2/p},
\end{align*}
which is exactly \cref{eq:conjectureballgen}.

\subsection{Proving $(2,p)$-hypercontractive inequality over $\Z_r$}
\label{sec:sec3.2}
Having proven the base case of our main theorem statement, we are now ready to prove our hypercontractivity theorem for matrix-valued functions over $\Z_r$.
\begin{theorem}
    \label{thm:genmathypercontractivity}
    Let $p\in [1,2]$. For every $f:\Z_r^n\rightarrow \mathbb{C}^{m\times m}$ and
    \begin{align*}
        \rho \leq \sqrt{\frac{(p-1)(1-(p-1)^{r-1})}{(r-1)(2-p)}},
    \end{align*}
    we have
    $$
    \left(\sum_{S\in \Z_r^n} \rho^{2|S|} \|\widehat{f}(S)\|_p^2\right)^{1/2} \leq \left(\frac{1}{r^n}\sum_{x\in\Z_r^n}\|f(x)\|_p^p\right)^{1/p},
    $$
    where $|S| \triangleq |\{i\in[n]:S_i\neq 0\}|$.
\end{theorem} 
\begin{proof}
For ease of notation, define $\zeta \triangleq \frac{(p-1)(1-(p-1)^{r-1})}{(r-1)(2-p)}$. It suffices to prove the inequality for $\rho = \sqrt{\zeta}$. Our proof closely follows the one in~\cite{ben2008hypercontractive} and is by induction on $n$. For $n=1$, the desired statement is
\begin{align}
    \label{eq:basecaseinduction}
    \sum_{S\in \Z_r} \zeta^{|S|} \|\widehat{f}(S)\|_p^2 &\leq  \left(\frac{1}{r}\sum_{x\in\Z_r}\|f(x)\|_p^p\right)^{2/p} \\
    &\Updownarrow\nonumber\\
    \|\widehat{f}(0) \|_p^2 + \zeta\sum_{k=1}^{r-1} \| \widehat{f}(k)\|_p^2 &\leq \left(\frac{1}{r}\sum_{k=0}^{r-1}\left\|\sum_{j=0}^{r-1}\omega_r^{j k}\widehat{f}(j)\right\|_p^p\right)^{2/p},\nonumber
\end{align}
where we used that $f(k) = \sum_{j=0}^{r-1}\omega_r^{j k}\widehat{f}(j)$. But this is precisely \cref{eq:conjectureballgen} with $A_k = \widehat{f}(k)$, $k\in\mathbb{Z}_r$.

We now assume the inequality holds for $n$ and prove it for $n+1$. Let $f:\Z_r^{n+1}\rightarrow \mathbb{C}^{m\times m}$ and $g_i=f\vert_{x_{n+1}=i}$ for $i\in \{0,\ldots,r-1\}$ be the function obtained by fixing the last {coordinate} of $f(\cdot)$ to $i$. By the induction hypothesis we have that, for every $i\in \{0,\ldots,r-1\}$ and $p\in [1,2]$,
\begin{align*}
    \sum_{S\in \Z_r^n} \zeta^{|S|} \|\widehat{g_i}(S)\|_p^2\leq  \left(\frac{1}{r^n}\sum_{x\in\Z_r^n}\|g_i(x)\|_p^p\right)^{2/p}.
\end{align*}
We now take the $\ell_p$ average of each of these $r$ inequalities to obtain
\begin{align}
    \label{eq:ellqaverage}
    \left(\frac{1}{r}\sum_{i=0}^{r-1}\left(\sum_{S\in \Z_r^n} \zeta^{|S|} \|\widehat{g_i}(S)\|_p^2\right)^{p/2}\right)^{2/p}&\leq  \left( \frac{1}{r}\sum_{i=0}^{r-1}\frac{1}{r^n}\sum_{x\in \Z_r^n}\|g_i(x)\|_p^p\right)^{2/p}
    &= \left(\frac{1}{r^{n+1}}\sum_{x\in \Z_r^{n+1}} \|f(x)\|_p^p\right)^{2/p}.
\end{align}
The right-hand side of the inequality above is exactly the right-hand side of the conjectured hypercontractive inequality. Below, we show how to lower bound the left-hand side of the above inequality by the desired left-hand side of the conjectured statement. To do so, we will need the following Minkowski's inequality.
\begin{lemma}[{Minkowski's inequality, \cite[Theorem 26]{hardy1952j}}]
    \label{lem:minkowski}
    For any $r_1\times r_2$ matrix whose rows are given by $u_1,\dots,u_{r_1}$ and whose columns are given by $v_1,\dots,v_{r_2}$, and any $1\leq q_1\leq q_2\leq \infty$,
    \begin{align*}
        \left\|\left(\|v_1\|_{q_2},\dots,\|v_{r_2}\|_{q_2}\right)\right\|_{q_1} \geq \left\|\left(\|u_1\|_{q_1},\dots,\|u_{r_1}\|_{q_1}\right)\right\|_{q_2}.
    \end{align*}
\end{lemma}
Now, consider the $r^n\times r$ matrix whose entries are given by
$
c_{S,i}=r^{n/2}\big\|\zeta^{|S|/2} \widehat{g_i}(S)\big\|_p
$
for every $i\in \{0,\ldots,r-1\}$ and $S\in \Z_r^n$. Then the left-hand side of \cref{eq:ellqaverage} can be written as
\begin{align}
    \displaybreak[0]\left(\frac{1}{r}\sum_{i=0}^{r-1}\left(\sum_{S\in \Z_r^n} \zeta^{|S|} \|\widehat{g_i}(S)\|_p^2\right)^{p/2}\right)^{1/p}&=    \left(\frac{1}{r}\sum_{i=0}^{r-1}\left(\frac{1}{r^n}\sum_{S\in \Z_r^n} c_{S,i}^2\right)^{p/2}\right)^{1/p}\nonumber\\\displaybreak[0]
    &\geq \left(\frac{1}{r^n}\sum_{S\in \Z_r^n}\left(\frac{1}{r}\sum_{i=0}^{r-1} c_{S,i}^p\right)^{2/p}\right)^{1/2}\tag{\cref{lem:minkowski} with $q_1 = p$ and $q_2=2$} \nonumber\\\displaybreak[0]
    &= \left(\sum_{S\in \Z_r^n}\zeta^{|S|}\left(\frac{1}{r}\sum_{i=0}^{r-1} \|\widehat{g_i}(S)\big\|_p^p\right)^{2/p}\right)^{1/2}.\label{eq:afterminkowski}
\end{align}
Now, for a fixed $S\in\Z_r^n$, we use the base case $n=1$, \cref{eq:basecaseinduction}, on the functions $h(i)=\widehat{g_i}(S)$ in order to get
$$
\left(\frac{1}{r}\sum_{i=0}^{r-1}\|\widehat{g_i}(S)\|_p^p\right)^{2/p}\geq \sum_{i=0}^{r-1} \zeta^{|i|} \left\|\frac{1}{r}\sum_{j=0}^{r-1}h(j)\omega_r^{-ij}\right\|_p^2 = \sum_{i=0}^{r-1} \zeta^{|i|} \left\|\frac{1}{r}\sum_{j=0}^{r-1}\widehat{g_j}(S)\omega_r^{-ij}\right\|_p^2.
$$
Plugging this back into \cref{eq:afterminkowski}, we have
\begin{align*}
    \left(\sum_{S\in \Z_r^n}\zeta^{|S|}\left(\frac{1}{r}\sum_{i=0}^{r-1} \|\widehat{g_i}(S)\|_p^p\right)^{2/p}\right)^{1/2} &\geq \left(\sum_{S\in \Z_r^{n}}\sum_{i=0}^{r-1}\zeta^{|S|+|i|}\left\|\frac{1}{r}\sum_{j=0}^{r-1}\widehat{g_j}(S)\omega_r^{-i j}\right\|_p^2\right)^{1/2} \\
    &= \left(\sum_{S\in \Z_r^{n+1}}\zeta^{|S|}\|\widehat{f}(S)\|_p^2\right)^{1/2},
\end{align*}
where we used the fact that $g_j=f\vert_{x_{n+1}=j}$, so, for every $i\in\Z_r$ and $S\in\Z_r^n$, we have that $\widehat{f}(S,i)=\frac{1}{r}\sum_{j=0}^{r-1}\widehat{g_j}(S)\omega_r^{-i j}$. The lower bound we obtained above is exactly the left-hand side of the conjectured hypercontractive inequality, which proves the theorem statement. 
\end{proof}
\begin{remark}[Comparison with hypercontractivity for real numbers]
    \label{rem:rem_realfunctions}
    For real functions $f:\Z_r^n\rightarrow \R$, it is known that~{\rm \cite{latala2000between,wolff2007hypercontractivity}} (see also~{\rm \cite[Theorem~10.18]{o2014analysis}})
    $$
    \left(\sum_{S\in \Z_r^n} \rho^{2|S|}|\widehat{f}(S)|^2\right)^{1/2} \leq \left(\frac{1}{r^n}\sum_{x\in\Z_r^n}|f(x)|^p\right)^{1/p},
    $$
    where $\rho \leq \sqrt{\frac{(r-1)^{1-1/p} - (r-1)^{-(1-1/p)}}{(r-1)^{1/p} - (r-1)^{-1/p}}}$. Moreover, this bound on $\rho$ is perfectly sharp, meaning that our bound $\rho \leq \sqrt{\frac{(p-1)(1-(p-1)^{r-1})}{(r-1)(2-p)}}$ in {\rm \cref{thm:genmathypercontractivity}} can possibly be improved. {We note that, for $r=2$, both bounds become $\rho \leq \sqrt{p-1}$.}
\end{remark}

\section{Hidden Hypermatching Problem}

In the following, an $\alpha$-partial $t$-hypermatching $M\in\mathcal{M}_{t,n}^\alpha$ on $n$ vertices is defined as a sequence of $\alpha n/t$ disjoint hyperedges $(M_{1,1},\dots,M_{1,t}),\dots,(M_{\alpha n/t, 1},\dots, M_{\alpha n/t, t})\in[n]^t$ with $t$ vertices each, where $\mathcal{M}_{t,n}^\alpha$ is the set of all such hypermatchings. If $\alpha = 1$, we shall write $\mathcal{M}_{t,n}$.

\begin{definition}
\label{def:hypermatching}
    Let $n,t\in\mathbb{N}$ be such that $t|n$ and $\alpha\in(0,1]$. In the $r$-ary Hidden Hypermatching $(r\text{-}\HH(\alpha,t,n))$ problem, Alice gets $x\in\Z_r^n$, Bob gets an $\alpha$-partial $t$-hypermatching $M\in\mathcal{M}_{t,n}^\alpha$ and a string $w\in\Z_r^{\alpha n/t}$. The hyperedges of $M$ are $(M_{1,1},\dots,M_{1,t}),\dots,(M_{\alpha n/t, 1},\dots, M_{\alpha n/t, t})$. Let $M\in\{0,1\}^{\alpha n/t \times n}$ also be the {incidence} matrix of Bob's hypermatching.  Consider the~distributions:
    %
    \begin{enumerate}
        \item  $\YES$ distribution $\mathcal{D}^{\YES}$, let $w=Mx$ (where the matrix product $Mx$ is over $\Z_r$);
        \item  $\NO$ distribution $\mathcal{D}^{\NO}$, $w$ is uniformly random in $\Z_r^{\alpha n/t}$.
    \end{enumerate}
    In the $r$-ary Hidden Hypermatching problem, Alice sends a message to Bob who needs to decide with high probability if $w$ is drawn from $\mathcal{D}^{\YES}$ or $\mathcal{D}^{\NO}$. 
\end{definition}

\subsection{Quantum protocol for $r$-ary Hidden Hypermatching}

For $t=2$, we obtain a efficient quantum communication protocol to solve the $r$-ary Hidden Hypermatching problem.
\begin{theorem}
    \label{thr:rary-upper}
    Given $\varepsilon > 0$, there is a protocol for the $r\text{-}\HH(\alpha,2,n)$ problem with one-sided error $\varepsilon$ and $O(\frac{1}{\alpha}\log{(nr)}\log(1/\varepsilon))$ qubits of communication from Alice to Bob.
\end{theorem}
\begin{proof}
    Let $M\in\mathcal{M}^\alpha_{2,n}$ be Bob's matching with disjoint edges $(M_{1,1},M_{1,2}),\dots,(M_{\alpha n/2,1},M_{\alpha n/2,2})\in[n]^2$. Alice first sends the following state {in $\mathbb{C}^r \otimes \mathbb{C}^{n}$} to Bob,
	\begin{align*}
		\frac{1}{\sqrt{n}}\sum_{i=1}^{n} |x_i,i\rangle.
	\end{align*}
	{Bob then} measures the state with the POVM $\{E_1,\dots,E_{\alpha n/2}, \mathbb{I}-\sum_{i=1}^{\alpha n/2} E_i\}$, where
	\begin{align*}
		E_i \triangleq |M_{i,1}\rangle\langle M_{i,1}| + |M_{i,2}\rangle\langle M_{i,2}|, \qquad \forall i\in [\alpha n/2]
	\end{align*}
	 {(note that $|M_{i,j}\rangle\in\mathbb{C}^n$ for $i\in[\alpha n/2],j\in[2]$)}. With probability $1-\alpha$ the POVM outputs the final outcome, and with probability $\alpha$ he will obtain a measurement outcome $E_i$ with $i\in[\alpha n/2]$ and get the state
	\begin{align*}
		|\psi\rangle \triangleq \frac{1}{\sqrt{2}}(|x_{M_{i,1}},M_{i,1}\rangle +  |x_{M_{i,2}},M_{i,2}\rangle).
	\end{align*}
	By repeating the procedure $O(1/\alpha)$ times, Bob obtains an outcome $i\in[\alpha n/2]$ with high probability.
	
	For the ease of notation, we can write $M_{i,1} = 0$ and $M_{i,2} = 1$ (note that Bob knows the values of both $M_{i,1},M_{i,2}$ explicitly). Bob now attaches a $\lceil\log_2{r}\rceil$-qubit register in the state $|0\rangle$ to $|\psi\rangle$ and applies a {quantum} Fourier transform $Q_r$ over $\Z_r$ to it to obtain
	\begin{align*}
		|0\rangle|\psi\rangle \to \frac{1}{\sqrt{r}}\sum_{k=0}^{r-1} |k\rangle|\psi\rangle.
	\end{align*}

	From now on we shall consider a parameter $\ell\in\mathbb{Z}_r$ to be determined later. Let $X$ be the usual Pauli operator and let $S_\ell$ and $P$ be the shift and phase operators over $\Z_r$ defined as $S_\ell|k\rangle = |\ell-k\rangle$ and $P|k\rangle = \omega_r^{k}|k\rangle$ for $k\in\mathbb{Z}_r$. Let $C_\ell \triangleq PS_\ell P\otimes X$. Bob applies the controlled unitary $U_\ell$ defined as $U_\ell|k\rangle|\psi\rangle = |k\rangle C_\ell^k|\psi\rangle$ on his state, followed by an inverse {quantum} Fourier transform $Q^\dagger_r$ on his first register to get
	\begin{align*}
		\frac{1}{\sqrt{r}}\sum_{k=0}^{r-1} U_\ell|k\rangle|\psi\rangle = \frac{1}{\sqrt{r}}\sum_{k=0}^{r-1} |k\rangle C_\ell^k|\psi\rangle \stackrel{Q^\dagger_r\otimes \mathbb{I}}{\longrightarrow} \frac{1}{r}\sum_{j=0}^{r-1}\sum_{k=0}^{r-1}\omega_r^{-jk} |j\rangle C_\ell^k|\psi\rangle.
	\end{align*}
	Let us calculate $C_\ell|\psi\rangle$ and $C_\ell^2|\psi\rangle$. We have
	\begin{align}
		C_\ell|\psi\rangle &= \frac{1}{\sqrt{2}}(PS_\ell P \otimes X)(|x_0,0\rangle +  |x_1,1\rangle)\nonumber\\
		&= \frac{1}{\sqrt{2}}(PS_\ell \otimes \mathbb{I})(\omega_r^{x_0}|x_0,1\rangle +  \omega_r^{x_1}|x_1,0\rangle)\nonumber\\
		&= \frac{1}{\sqrt{2}}(P \otimes \mathbb{I})(\omega_r^{x_0}|\ell -x_0,1\rangle +  \omega_r^{x_1}|\ell -x_1,0\rangle)\nonumber\\
		&= \frac{\omega_r^{\ell}}{\sqrt{2}}(|\ell -x_1,0\rangle + |\ell -x_0,1\rangle)\label{eq:bobstate}
	\end{align}
    and
	\begin{align*}
	    C_\ell^2|\psi\rangle &= \frac{\omega_r^{\ell}}{\sqrt{2}}(PS_\ell P \otimes X)(|\ell -x_1,0\rangle + |\ell -x_0,1\rangle)\\
	    &= \frac{\omega_r^{\ell}}{\sqrt{2}}(PS_\ell \otimes \mathbb{I})(\omega_r^{\ell - x_1}|\ell -x_1,1\rangle + \omega_r^{\ell - x_0}|\ell -x_0,0\rangle)\\
	    &= \frac{\omega_r^{\ell}}{\sqrt{2}}(P \otimes \mathbb{I})(\omega_r^{\ell - x_1}|x_1,1\rangle + \omega_r^{\ell - x_0}|x_0,0\rangle)\\
	    &= \omega_r^{2\ell}|\psi\rangle.
	\end{align*}
	We can see from the above that $C_\ell^{2k}|\psi\rangle = \omega_r^{2k\ell}|\psi\rangle$. By defining $\Delta_\ell \triangleq \ell - (x_0+x_1)$ and $\delta_k =1$ if $k$ is odd and $0$ otherwise, Bob's final state is
	\begin{align}
		\frac{1}{r}\sum_{j=0}^{r-1}\sum_{k=0}^{r-1}\omega_r^{k(\ell-j)}|j\rangle\frac{1}{\sqrt{2}}(|x_0 + \Delta_\ell\delta_k,0\rangle + |x_1 + \Delta_\ell\delta_k,1\rangle).\label{eq:bobfinal}
	\end{align}
	Now observe that, if $\ell=x_0+x_1$, then $C_\ell|\psi\rangle = \omega_r^{\ell}|\psi\rangle$ in \cref{eq:bobstate}. This means that Bob's state in \cref{eq:bobfinal} becomes $|x_0+x_1\rangle|\psi\rangle$, and if he measures his first register, he obtains $x_0+x_1~(\text{mod}~r)$ with certainty.
	
	On the other hand, if $\ell\neq x_0+x_1$, then the probability of measuring the first register and obtaining the outcome $m\in\mathbb{Z}_r$ is
	\begin{align*}
	    \operatorname{Pr}[m] &= \frac{1}{2r^2}\sum_{k_1,k_2=0}^{r-1}\omega_r^{(\ell-m)(k_1-k_2)}(\langle x_0 + \Delta_\ell\delta_{k_2}|x_0 + \Delta_\ell\delta_{k_1}\rangle + \langle x_1 + \Delta_\ell\delta_{k_2}|x_1 + \Delta_\ell\delta_{k_1}\rangle)\\
		&= \frac{1}{r^2}\sum_{k_1,k_2~\text{even}}^{r-1}\omega_r^{(\ell-m)(k_1-k_2)} + \frac{1}{r^2}\sum_{k_1,k_2~\text{odd}}^{r-1}\omega_r^{(\ell-m)(k_1-k_2)}\\
		&= \left|\frac{1}{r}\sum_{k~\text{even}}^{r-1}\omega_r^{k(\ell-m)}\right|^2 + \left|\frac{1}{r}\sum_{k~\text{odd}}^{r-1}\omega_r^{k(\ell-m)}\right|^2.
	\end{align*}
    It is not hard to see that the above probability is {maximal} when $m = \ell$, in which case
	\begin{align*}
		\operatorname{Pr}[m=\ell] = \frac{1}{r^2}\left\lfloor\frac{r+1}{2}\right\rfloor^2 + \frac{1}{r^2}\left\lfloor\frac{r}{2}\right\rfloor^2 = \begin{cases}
			\frac{1}{2} \qquad &r~\text{even},\\
			\frac{1}{2} + \frac{1}{2r^2} \qquad &r~\text{odd}.
		\end{cases}
	\end{align*}
	Given the considerations above, Bob uses the following strategy: he picks $\ell$ as the corresponding entry $w_i$ from $w\in\mathbb{Z}_r^{\alpha n/2}$ given the measured hyperedge $(M_{i,1},M_{i,2})$. If the outcome $m$ from measuring his final state in \cref{eq:bobfinal} equals $w_i$, then he outputs $\mathsf{YES}$, otherwise he outputs $\mathsf{NO}$. Indeed, in the $\mathsf{YES}$ instance, $w_i = x_{M_{i,1}}+x_{M_{i,2}}$ and so $m$ equals $w_i$ with probability $1$, while in the $\mathsf{NO}$ instance, $m$ equals $w_i$ with probability at most $\frac{1}{2}+\frac{1}{2r^2}$. Thus the communication protocol has one-sided error at most $\frac{1}{2}+\frac{1}{2r^2}$, i.e., $\operatorname{Pr}[\text{error}|\mathsf{YES}] = 0$ and $\operatorname{Pr}[\text{error}|\mathsf{NO}] \leq \frac{1}{2}+\frac{1}{2r^2}$. By repeating the whole protocol $O(\log(1/\varepsilon))$ more times, the one-sided error probability can be decreased to $\varepsilon$: if in any of the repetitions the final measurement outcome is different from $w_i$, then Bob knows that $\mathsf{NO}$ is the correct answer.
\end{proof}

\subsection{Quantum lower bound on $r$-ary Hidden Hypermatching}

In this section we shall turn our attention to proving quantum and classical lower bounds on the amount of communication required by the $r$-$\HH(\alpha,t,n)$ problem, but first we need the following~lemma.
\begin{lemma}
    \label{lem:hypercontractive}
    Let $f:\Z_r^n \to \operatorname{D}(\mathbb{C}^{2^m})$ be any mapping from an $n$-bit alphabet to $m$-qubit density matrices. Then for any $\delta\in[0,1/(r-1)]$, we have
    \begin{align*}
        \sum_{S\in\Z_r^n} \delta^{|S|}\|\widehat{f}(S)\|^2_1 \leq 2^{2(r-1)\delta m}.
    \end{align*}
\end{lemma}
\begin{proof}
    Let $p \triangleq 1+(r-1)\delta$. First note that, given the eigenvalues $\sigma_1,\dots,\sigma_{2^m}$ from $f(x)$, which are non-negative reals that sum to $1$, we have
    \begin{align*}
        \|f(x)\|_p^p = \sum_{i=1}^{2^m}\sigma_i^p \leq \sum_{i=1}^{2^m}\sigma_i = 1.
    \end{align*}
    Using \cref{thm:genmathypercontractivity} and \cref{rem:1}, we now get
    \begin{align*}
        \sum_{S\in\Z_r^n}\left(\frac{p-1}{r-1}\right)^{|S|}\|\widehat{f}(S)\|_p^2 \leq \left(\frac{1}{r^n}\sum_{x\in\Z_r^n}\|f(x)\|_p^p\right)^{2/p} \leq \left(\frac{1}{r^n}\cdot r^n\right)^{2/p} = 1.
    \end{align*}
    On the other hand, note $p \leq q \implies \big(\frac{1}{2^m}\sum_{i=1}^{2^m}\sigma_i^p\big)^{1/p} \leq \big(\frac{1}{2^m}\sum_{i=1}^{2^m}\sigma_i^q\big)^{1/q}$ by H\"older's inequality, hence
    \begin{align*}
        \sum_{S\in\Z_r^n}\left(\frac{p-1}{r-1}\right)^{|S|}2^{-2m/p}\|\widehat{f}(S)\|_p^2 \geq \sum_{S\in\Z_r^n}\left(\frac{p-1}{r-1}\right)^{|S|}2^{-2m}\|\widehat{f}(S)\|_1^2.
    \end{align*}
    Rearranging the inequalities leads to
    \[
        \sum_{S\in\Z_r^n}\left(\frac{p-1}{r-1}\right)^{|S|}\|\widehat{f}(S)\|_1^2 \leq 2^{2m(1-1/p)} \leq 2^{2m(p-1)}. \qedhere
    \]
\end{proof}

We are now ready to state and prove our main quantum communication complexity lower bound for the $r$-ary Hidden Hypermatching problem. 
\begin{theorem}
\label{thm:hypermatchinglowerbound}
    Any quantum protocol that achieves advantage $\varepsilon>0$ for the $r\text{-}\HH(\alpha,t,n)$ problem with $t\geq 3$ and $\alpha \leq \min(1/2, (r-1)^{-1/2})$ requires $\Omega(r^{-(1+1/t)}(\varepsilon^2/\alpha)^{2/t}(n/t)^{1-2/t})$ qubits of communication from Alice to Bob.
\end{theorem}
Notice that for $r=2$ our lower bound reads $\Omega(\alpha^{-2/t}(n/t)^{1-2/t})$, which has a better dependence on $\alpha$ compared to the lower bound $\Omega(\log(1/\alpha)(n/t)^{1-2/t})$ from~\cite{shi2012limits}. Also, see \cref{rem:alphadependence} at the end of the section for an improvement on the requirement $\alpha \leq \min(1/2, (r-1)^{-1/2})$.
\begin{proof}
    Consider an $m$-qubit communication protocol. An arbitrary $m$-qubit protocol can be viewed as Alice sending an encoding of her input $x\in \Z_r^n$ into a quantum state so that Bob can distinguish if his $w$ was drawn from $\mathcal{D}^{\YES}$ or $\mathcal{D}^{\NO}$. 
    Let $\rho:\Z_r^n\to \operatorname{D}(\mathbb{C}^{2^m})$ be Alice's encoding function. For our `hard' distribution, Alice and Bob receive $x\in\Z_r^n$ and $M\in\mathcal{M}_{t,n}^\alpha$, respectively, uniformly at random, while Bob's input $w\in\Z_r^{\alpha n/t}$ is drawn from the distribution $\mathcal{D} \triangleq \frac{1}{2}\mathcal{D}^{\YES} + \frac{1}{2}\mathcal{D}^{\NO}$, i.e., with probability $1/2$ is comes from $\mathcal{D}^{\YES}$, and with probability $1/2$ it comes from $\mathcal{D}^{\NO}$. Let $p_x \triangleq r^{-n}$, $p_M \triangleq |\mathcal{M}^{\alpha}_{t,n}|^{-1}$ and $p_w \triangleq r^{-\alpha n/t}$, then our hard distribution $\mathcal{P}$ is
	\begin{align}
	\label{eq:eq3.5.c3}
		\operatorname{Pr}[x,\YES,M,w] = \frac{1}{2}p_x\cdot p_M \cdot [Mx = w], \qquad 		\operatorname{Pr}[x,\NO,M,w] = \frac{1}{2}p_x\cdot p_M\cdot  p_w.
	\end{align}
	    
	Conditioning on Bob's input $(M, w)$, from his perspective, Alice sends the message $\rho(x)$ with probability $\operatorname{Pr}[x|M,w]$. Therefore, conditioned on an instance of the problem ($\YES$ or $\NO$), Bob receives one of the following two quantum states $\rho_\YES^{M,w}$ and $\rho_\NO^{M,w}$, each appearing with probability $\operatorname{Pr}[\YES|M,w]$ and $\operatorname{Pr}[\NO|M,w]$, respectively,
	\begin{align}
	\label{eq:rhoyesrhono}
	\begin{aligned}
		\rho_{\YES}^{M,w} &= \sum_{x\in\Z_r^n} \operatorname{Pr}[x|\YES,M,w]\cdot\rho(x) = \frac{1}{\operatorname{Pr}[\YES,M,w]} \sum_{x\in\Z_r^n} \operatorname{Pr}[x,\YES,M,w] \cdot \rho(x),\\ 
	    \rho_{\NO}^{M,w} &= \sum_{x\in\Z_r^n} \operatorname{Pr}[x|\NO,M,w]\cdot\rho(x) = \frac{1}{\operatorname{Pr}[\NO,M,w]} \sum_{x\in\Z_r^n} \operatorname{Pr}[x,\NO,M,w] \cdot \rho(x).
	\end{aligned}
	\end{align}
	Bob's best strategy to determine the distribution of $w$ conditioning on his input $(M,w)$ is no more than the chance to distinguish between these two quantum states $\rho_\YES^{M,w}$ and $\rho_\NO^{M,w}$.
	
	Now let $\varepsilon_{\rm bias}$ be the bias of the protocol that distinguishes between $\rho_\YES^{M,w}$ and $\rho_\NO^{M,w}$. According to \cref{lem:lem3.5.c3}, the bias $\varepsilon_{\rm bias}$ of any quantum protocol for a fixed $M$ and $w$ can be upper bounded as
	\begin{align*}
		\varepsilon_{\rm bias} \leq \big\|{\operatorname{Pr}}[\YES|M,w]\cdot\rho_\YES^{M,w} - \operatorname{Pr}[\NO|M,w]\cdot\rho_\NO^{M,w}\big\|_1.
	\end{align*}
	We prove in \cref{thr:raryhidden} below that, if $m \leq \frac{\gamma}{r^{1+1/t}}(\frac{\varepsilon^2}{\alpha})^{2/t} (n/t)^{1-2/t}$ for a universal constant $\gamma$, then the average bias over $M$ and $w$ is at most $\varepsilon^2$, i.e.,
	\begin{align*}
		\operatorname*{\mathbb{E}}_{(M,w)\sim\mathcal{P}_{M,w}}[\varepsilon_{\rm bias}] \leq \varepsilon^2,
	\end{align*}
	where $\mathcal{P}_{M,w}$ is the marginal distribution of $\mathcal{P}$. Therefore, by Markov's inequality, for at least a $(1-\varepsilon)$-fraction of $M$ and $w$, the bias in distinguishing between $\rho_\YES^{M,w}$ and $\rho_\NO^{M,w}$ is $\varepsilon$-small. Therefore, Bob's advantage over randomly guessing the right distribution will be at most $\varepsilon$ (for the event that $M$ and $w$ are such that the distance between $\rho_\YES^{M,w}$ and $\rho_\NO^{M,w}$ is more than $\varepsilon$) plus $\varepsilon/2$ (for the advantage over random guessing when $\varepsilon_{\rm bias} \leq \varepsilon$), and so $m = \Omega(r^{-(1+1/t)}(\varepsilon^2/\alpha)^{2/t}(n/t)^{1-2/t})$.
\end{proof}

\begin{theorem}
	\label{thr:raryhidden}
    For $x\in\Z_r^n$, $M\in\mathcal{M}_{t,n}^\alpha$, $w\in\Z_r^{\alpha n/t}$ and $b\in\{\YES,\NO\}$, consider the probability distribution $\mathcal{P}$ defined in {\rm \cref{eq:eq3.5.c3}}. Given an encoding function $\rho:\Z_r^n\to \operatorname{D}(\mathbb{C}^{2^m})$, consider the quantum states $\rho_\YES^{M,w}$ and $\rho_\NO^{M,w}$ from {\rm \cref{eq:rhoyesrhono}}. If $\alpha \leq \min(1/2,(r-1)^{-1/2})$, then there is a universal constant $\gamma>0$ (independent of $n$, $t$, $r$ and $\alpha$), such that, if $m \leq \frac{\gamma}{r^{1+1/t}}(\frac{\varepsilon^2}{\alpha})^{2/t} (n/t)^{1-2/t}$ for all $\varepsilon > 0$, then
    \begin{align*}
        \operatorname*{\mathbb{E}}_{(M,w)\sim\mathcal{P}_{M,w}}\left[\big\|\operatorname{Pr}[\YES|M,w]\cdot\rho_\YES^{M,w} - \operatorname{Pr}[\NO|M,w]\cdot\rho_\NO^{M,w}\big\|_1\right] \leq \varepsilon^2.
    \end{align*}
\end{theorem}
\begin{proof}
	For the ease of notation, we shall write
	\begin{align*}
	    \varepsilon_{\rm bias} \triangleq \operatorname*{\mathbb{E}}_{(M,w)\sim\mathcal{P}_{M,w}}\left[\big\|\operatorname{Pr}[\YES|M,w]\cdot\rho_\YES^{M,w} - \operatorname{Pr}[\NO|M,w]\cdot\rho_\NO^{M,w}\big\|_1\right].
	\end{align*}
	Therefore, we have that
	\begin{align*}
		\displaybreak[0]\varepsilon_{\rm bias} &= \sum_{M\in\mathcal{M}_{t,n}^\alpha}\sum_{w\in\Z_r^{\alpha n/t}}\operatorname{Pr}[M,w]\cdot \big\|\operatorname{Pr}[\YES|M,w]\cdot\rho_\YES^{M,w} - \operatorname{Pr}[\NO|M,w]\cdot\rho_\NO^{M,w}\big\|_1\displaybreak[0]\\
		&= \sum_{M\in\mathcal{M}_{t,n}^\alpha}\sum_{w\in\Z_r^{\alpha n/t}}\Big\|\sum_{x\in\Z_r^n}\Big(\operatorname{Pr}[x,\YES,M,w]\cdot\rho(x) - \operatorname{Pr}[x,\NO,M,w]\cdot\rho(x)\Big)\Big\|_1\displaybreak[0] \tag{\cref{eq:rhoyesrhono} and conditional probability}\\
		&= \sum_{M\in\mathcal{M}_{t,n}^\alpha}\sum_{w\in\Z_r^{\alpha n/t}}\Big\|\sum_{x\in\Z_r^n}\frac{1}{2}p_x\cdot p_M \left(\big[Mx = w\big] - p_w\right)\rho(x)\Big\|_1 && \tag{\cref{eq:eq3.5.c3}}\displaybreak[0]\\
		&= \sum_{M\in\mathcal{M}_{t,n}^\alpha}\sum_{w\in\Z_r^{\alpha n/t}}\Big\|\sum_{x\in\Z_r^n}\frac{1}{2}p_x\cdot p_M \left(\big[Mx = w\big] - p_w\right)\cdot \sum_{S\in\Z_r^n}\widehat{\rho}(S)\omega_r^{S\cdot x}\Big\|_1\tag{Fourier decomposition of $\rho$}\displaybreak[0]\\
		&= \sum_{M\in\mathcal{M}_{t,n}^\alpha}\sum_{w\in\Z_r^{\alpha n/t}} \Big\|\sum_{S\in\Z_r^n} u(M,w,S) \widehat{\rho}(S)\Big\|_1\displaybreak[0]\\
		&\leq \sum_{S\in\Z_r^n}\sum_{M\in\mathcal{M}_{t,n}^\alpha}\sum_{w\in\Z_r^{\alpha n/t}}|u(M,w,S)|\cdot \|\widehat{\rho}(S)\|_1,
	\end{align*}
	where we defined
	\begin{align}
		\label{eq:defnofu}
			u(M,w,S) \triangleq \frac{1}{2}\sum_{x\in\Z_r^n}p_x\cdot p_M \cdot \omega_r^{S\cdot x} \left(\big[Mx = w\big] - p_w\right).
		\end{align}
	
    Next, we upper bound the quantity $u(M,w,S)$ using the lemma below. {In the following lemma, let $I(M) = \bigcup_{i\in[\alpha n/t],j\in[t]} \{M_{i,j}\}$ be the set of indices of the $\alpha$-partial matching $M\in\mathcal{M}_{n,t}^\alpha$.} Moreover, we shall write $S|_{M_i} = S_{M_{i,1}} S_{M_{i,2}}\dots S_{M_{i,t}}\in\Z_r^t$ to denote the string $S$ restricted to the hyperedge $M_i = (M_{i,1},\dots,M_{i,t})$, where $S_{M_{i,j}}$ is the $M_{i,j}$-th entry of $S$. The same applies to $x\in\Z_r^n$.
    \begin{lemma}
    \label{lem:understandingu}
        Let $M\in\mathcal{M}^{\alpha}_{t,n}$ and {$I(M) = \bigcup_{i\in[\alpha n/t],j\in[t]} \{M_{i,j}\}$.} Define the set 
        \begin{align*}
            \Delta(M) = \{S\in\Z_r^n\setminus\{0^n\} ~|~ S_{M_{i,1}} = S_{M_{i,2}} = \cdots = S_{M_{i,t}}  ~\forall i\in[\alpha n/t] ~\text{and}~ S_j = 0  ~\forall j\notin I(M)\}.
        \end{align*}
        Given $u(M,w,S)$ as defined in {\rm \cref{eq:defnofu}} for $w\in\Z_r^{\alpha n/t}$ and $S\in\Z_r^n$, we have
        $
        u(M,w,S)=\frac{1}{2}\cdot r^{-\alpha n/t}\cdot p_M
        $
        if $S\in\Delta(M)$ and 0 if $S\notin\Delta(M)$.
    \end{lemma}
	
    \begin{proof}
        Recall the definition of $u$:
        $$
        u(M,w,S) = \frac{1}{2}\sum_{x\in\Z_r^n}p_x\cdot p_M \cdot \omega_r^{S\cdot x} \left(\big[Mx = w\big] - p_w\right).
        $$
        In order to understand this expression, we start with the following:
	\begin{align*}
            \displaybreak[0]\sum_{x\in\Z_r^n}\omega_r^{S\cdot x}\big[Mx = w\big] &= \sum_{x\in\Z_r^n}\omega_r^{S\cdot x} \prod_{i=1}^{\alpha n/t} \left[(Mx)_i = w_i\right]\displaybreak[0]\\
            &= \sum_{x\in\Z_r^n}\omega_r^{S\cdot x} \prod_{i=1}^{\alpha n/t} \left[\sum_{j=1}^t x_{M_{i,j}} = w_i\right]\displaybreak[0]\\
            &= \sum_{x\in\Z_r^{n}}\left(\prod_{j\notin I(M)}\omega_r^{S_jx_j}\right)\left(\prod_{i=1}^{\alpha n/t}\omega_r^{S|_{M_{i}}\cdot x|_{M_i}}\left[\sum_{j=1}^t x_{M_{i,j}} = w_i\right] \right)\displaybreak[0]\\
            &= \left(\prod_{j\notin I(M)}\sum_{x\in\Z_r}\omega_r^{S_jx}\right)\left(\prod_{i=1}^{\alpha n/t}\sum_{x\in\Z_r^{t}}\omega_r^{S|_{M_{i}}\cdot x}\left[\sum_{j=1}^t x_j = w_i\right] \right)\displaybreak[0]\tag{Relabelling $x\in\mathbb{Z}_r^n$}\\
            &= r^{n(1-\alpha)}[S_j=0 ~\forall j\notin I(M)]\prod_{i=1}^{\alpha n/t}\sum_{x\in\Z_r^t}\omega_r^{S|_{M_i}\cdot x}\left[\sum_{j=1}^t x_j = w_i\right].
	\end{align*}
         Now we use that $\sum_{j=1}^t x_j = w_i \implies x_t = w_i - \sum_{j=1}^{t-1} x_j$ modulo $r$, so that
	\begin{align*}
		S|_{M_i}\cdot x = \sum_{j=1}^t S_{M_{i,j}}x_j = \sum_{j=1}^{t-1} S_{M_{i,j}}x_j + S_{M_{i,t}}\left(w_i - \sum_{j=1}^{t-1}x_j\right) = S_{M_{i,t}}w_i + \sum_{j=1}^{t-1}(S_{M_{i,j}} - S_{M_{i,t}})x_j
	\end{align*}
	modulo $r$. This leads to {(assuming that $S_j = 0$ for all $j\notin I(M)$)}
	\begin{align}
	    \sum_{x\in\Z_r^n}\omega_r^{S\cdot x}\big[Mx = w\big] &= r^{n(1-\alpha)}\prod_{i=1}^{\alpha n/t}\omega_r^{S_{M_{i,t}}w_i}\sum_{x\in\Z_r^{t-1}}\omega_r^{\sum_{j=1}^{t-1}(S_{M_{i,j}} - S_{M_{i,t}})x_j}= \frac{r^{n}}{r^{\alpha n/t}}\prod_{i=1}^{\alpha n/t}\omega_r^{S_{M_{i,t}}w_i}\label{eq:rary1}
	\end{align}
	if, for all $i\in[\alpha n/t]$, $S_{M_{i,j}}$ is constant for all $j\in[t]$, i.e., if $S_{M_{i,1}} = S_{M_{i,2}} = \dots = S_{M_{i,t}}$ for all $i\in[\alpha n/t]$. Otherwise the above expression is $0$. 
	Thus, if $S_{M_{i,1}} = S_{M_{i,2}} = \dots = S_{M_{i,t}}$ for all $i\in[\alpha n/t]$ and $S_j=0$ for $j\notin I(M)$, then we can use \cref{eq:rary1} to get (remember that $p_x \triangleq r^{-n}$ and $p_w \triangleq r^{-\alpha n/t}$)
	\begin{align*}
   		\displaybreak[0]|u(M,w,S)| = \frac{1}{2}\left|\sum_{x\in\Z_r^n}p_xp_M \omega_r^{S\cdot x} \left(\big[Mx = w\big] - p_w\right)\right| &= \frac{p_M}{2r^{\alpha n/t}}\left|\prod_{i=1}^{\alpha n/t}\omega_r^{S_{M_{i,t}}w_i} - [S=0^n]\right|\displaybreak[0] \\
		&= \begin{cases}
			0 &~\text{if}~ S=0^n,\\
			\frac{1}{2} r^{-\alpha n/t} p_M &~\text{if}~S\neq 0^n.
		\end{cases}
	\end{align*}
	Hence, we have
	\begin{align*}
	    |u(M,w,S)| = \begin{cases}
			0 &\text{if}~ S=0^n,\\
			\frac{1}{2} r^{-\alpha n/t} p_M &\text{if}~S_{M_{i,1}} = S_{M_{1,2}} = \dots= S_{M_{i,t}}~ \forall i\in[\alpha n/t] ~\text{and}~ S_j = 0~\forall j\notin I(M),\\
			0 &\text{otherwise},
		\end{cases}
	\end{align*}
	proving the lemma statement
    \end{proof}
    	We now proceed to upper bound $\varepsilon_{\rm bias}$ using the expression for $|u(M,w,S)|$ from \cref{lem:understandingu}. For $S\in\Z_r^n$, let $|S| \triangleq |\{i\in[n]: S_i\neq 0\}|$. Notice that, if $S\in\Delta(M)$, then $|S| = kt$ for some $k\in[\alpha n/t]$. Hence, we have that
	\begin{align*}
		\varepsilon_{\rm bias} \leq \frac{1}{2}\sum_{S\in\Z_r^n}\sum_{\substack{M\in\mathcal{M}^\alpha_{t,n} \\ S\in\Delta(M)}}p_M \sum_{w\in\Z_r^{\alpha n/t}}\frac{1}{r^{\alpha n/t}}\|\widehat{\rho}(S)\|_1 &= \frac{1}{2}\sum_{k=1}^{\alpha n/t}\sum_{\substack{S\in\Z_r^n \\ |S| = kt}}\sum_{\substack{M\in\mathcal{M}^\alpha_{t,n} \\ S\in\Delta(M)}} p_M\|\widehat{\rho}(S)\|_1\\
		&= \frac{1}{2}\sum_{k=1}^{\alpha n/t}\sum_{\substack{S\in\Z_r^n \\ |S| = kt}}\operatorname*{Pr}_{M\sim\mathcal{M}_{t,n}^\alpha}[S\in\Delta(M)]\cdot\|\widehat{\rho}(S)\|_1,
	\end{align*}
	using that 
	\begin{align*}
	    \sum_{\substack{M\in\mathcal{M}^\alpha_{t,n}: S\in\Delta(M)}} p_M = \operatorname*{Pr}_{M\sim\mathcal{M}^{\alpha}_{t,n}}[S\in\Delta(M)].
	\end{align*}
	We now upper bound this probability using the following lemma.  
	\begin{lemma}
	    \label{lem:probcomb}
	    Let $t\in\Z$. Let $S\in\Z_r^n$ with $k_j \triangleq \frac{1}{t}\cdot |\{i\in[n]:S_i = j\}|\in\Z$ 
	    for $j\in\{1,\dots,r-1\}$. Let $k \triangleq \sum_{j=1}^{r-1} k_j$. For any $M\in\mathcal{M}_{t,n}^\alpha$, let $\Delta(M)$ be the set from {\rm \cref{lem:understandingu}}. Then
	    \begin{align*}
	        \operatorname*{Pr}_{M\sim\mathcal{M}_{t,n}^\alpha}[S\in\Delta(M)] = \frac{\binom{\alpha n/t}{k}}{\binom{n}{kt}}\frac{k!}{(kt)!}\prod_{j=1}^{r-1}\frac{(k_jt)!}{k_j!}.
	    \end{align*}
	\end{lemma}
	\begin{proof}
	    We can assume without loss of generality that $S = 1^{k_1t}2^{k_2t}\dots (r-1)^{k_{r-1}t}0^{n-kt}$. First note that the total number $|\mathcal{M}_{t,n}^\alpha|$ of $\alpha$-partial hypermatchings is $n!/\big((t!)^{\alpha n/t}(\alpha n/t)!(n-\alpha n)!\big)$. This can be seen as follows: pick a permutation of $n$, view the first $\alpha n/t$ tuples of length $t$ as $\alpha n/t$ hyperedges, and ignore the ordering within each hyperedge, the ordering of the $\alpha n/t$ hyperedges and the ordering of the last $n-\alpha n$ vertices. Now, given our particular $S$, notice that $S\in\Delta(M)$ if, for $j\in[r-1]$, $M$ has exactly $k_j$ hyperedges in 
	    $$
	    \left\{1 + t\sum_{i=1}^{j-1}k_i,~2 + t\sum_{i=1}^{j-1}k_i,~3 + t\sum_{i=1}^{j-1}k_i,\ldots,(k_j-1) + t\sum_{i=1}^{j-1}k_i,~t\sum_{i=1}^{j}k_i\right\},
	    $$ 
	    i.e., $k_1$ hyperedges in $\{1,\dots,k_1t\}$, $k_2$ hyperedges in $\{k_1t+1,\dots,(k_2+k_1)t\}$, etc., and also $\alpha n/t-k$ hyperedges in $[n]\setminus[kt]$. The number of ways to pick $k_j$ hyperedges in $\left\{1 + t\sum_{i=1}^{j-1}k_i,\dots,t\sum_{i=1}^{j}k_i\right\}$ is $(k_jt)!/((t!)^{k_j}k_j!)$. The number of ways to pick the remaining $\alpha n/t - k$ hyperedges in $[n]\setminus[kt]$ is $(n-kt)!/((t!)^{\alpha n/t - k} (\alpha n/t - k)! (n-\alpha n)!)$. Hence $\operatorname*{Pr}_{M\sim\mathcal{M}_{t,n}^\alpha}[S\in\Delta(M)]$ equals
	    \[
	        \frac{\frac{(n-kt)!}{(t!)^{\alpha n/t - k} (\alpha n/t - k)! (n-\alpha n)!}}{\frac{n!}{(t!)^{\alpha n/t}(\alpha n/t)!(n-\alpha n)!}}\prod_{j=1}^{r-1}\frac{(k_jt)!}{(t!)^{k_j}k_j!} = \frac{(n-kt)!(\alpha n/t)!}{n!(\alpha n/t - k)!}\prod_{j=1}^{r-1}\frac{(k_jt)!}{k_j!} = \frac{\binom{\alpha n/t}{k}}{\binom{n}{kt}}\frac{k!}{(kt)!}\prod_{j=1}^{r-1}\frac{(k_jt)!}{k_j!}. \qedhere
        \]
	\end{proof}
	Using \cref{lem:probcomb} and the notation $|S|_i \triangleq |\{j\in[n]:S_j=i\}|$, we continue bounding $\varepsilon_{\rm bias}$ as
	\begin{align*}
		\displaybreak[0]\varepsilon_{\rm bias} &\leq \frac{1}{2}\sum_{k=1}^{\alpha n/t}\frac{\binom{\alpha n/t}{k}}{\binom{n}{kt}}\sum_{\substack{k_1,\dots,k_{r-1}\geq 0 \\ \sum_{j=1}^{r-1} k_j = k}}\sum_{\substack{S\in\Z_r^n \\ |S|_i = k_it, ~i\in[r-1]}}\frac{k!}{(kt)!}\left(\prod_{j=1}^{r-1}\frac{(k_jt)!}{k_j!}\right) \|\widehat{\rho}(S)\|_1\displaybreak[0]\\
		&\leq \frac{1}{2}\sum_{k=1}^{\alpha n/t}\frac{\binom{\alpha n/t}{k}}{\binom{n}{kt}}\sqrt{\sum_{\substack{k_1,\dots,k_{r-1}\geq 0 \\ \sum_{j=1}^{r-1} k_j = k}}\sum_{\substack{S\in\Z_r^n \\ |S|_i = k_it, ~i\in[r-1]}}\frac{k!^2}{(kt)!^2}\prod_{j=1}^{r-1}\frac{(k_jt)!^2}{k_j!^2}} \sqrt{\sum_{\substack{k_1,\dots,k_{r-1}\geq 0 \\ \sum_{j=1}^{r-1} k_j = k}}\sum_{\substack{S\in\Z_r^n \\ |S|_i = k_it, ~i\in[r-1]}}\|\widehat{\rho}(S)\|_1^2}\tag{Cauchy-Schwarz}\displaybreak[0]\\
		&\leq \frac{1}{2}\sum_{k=1}^{\alpha n/t}\frac{\binom{\alpha n/t}{k}}{\binom{n}{kt}}\sqrt{\sum_{\substack{k_1,\dots,k_{r-1}\geq 0 \\ \sum_{j=1}^{r-1} k_j = k}}\sum_{\substack{S\in\Z_r^n \\ |S|_i = k_it, ~i\in[r-1]}}\frac{k!^2}{(kt)!^2}\prod_{j=1}^{r-1}\frac{(k_jt)!^2}{k_j!^2}} \sqrt{\sum_{\substack{S\in\Z_r^n \\ |S| = kt}}\|\widehat{\rho}(S)\|_1^2}\displaybreak[0]\\
		&= \frac{1}{2}\sum_{k=1}^{\alpha n/t}\frac{\binom{\alpha n/t}{k}}{\sqrt{\binom{n}{kt}}}\sqrt{\sum_{\substack{k_1,\dots,k_{r-1}\geq 0 \\ \sum_{j=1}^{r-1} k_j = k}} \frac{k!^2}{(kt)!}\prod_{j=1}^{r-1}\frac{(k_jt)!}{k_j!^2}} \sqrt{\sum_{\substack{S\in\Z_r^n \\ |S| = kt}}\|\widehat{\rho}(S)\|_1^2},
	\end{align*}
	where the last equality uses $\sum_{\substack{S\in\Z_r^n:|S|_i = k_it}} 1 = \binom{n}{kt}(kt)!\prod_{j=1}^{r-1}\frac{1}{(k_jt)!}$. Since
	\begin{align}
	    \label{eq:alphadependence}
		\sum_{\substack{k_1,\dots,k_{r-1}\geq 0 \\ \sum_{j=1}^{r-1} k_j = k}} \frac{k!^2}{(kt)!}\prod_{j=1}^{r-1}\frac{(k_jt)!}{k_j!^2} &= \sum_{\substack{k_1,\dots,k_{r-1}\geq 0 \\ \sum_{j=1}^{r-1} k_j = k}} \frac{\binom{k}{k_1,\dots,k_{r-1}}^2}{\binom{kt}{k_1t,\dots,k_{r-1}t}} \leq \sum_{\substack{k_1,\dots,k_{r-1}\geq 0 \\ \sum_{j=1}^{r-1} k_j = k}} \binom{k}{k_1,\dots,k_{r-1}} = (r-1)^k,
	\end{align}
	where the last equality follows from the the multinomial theorem, then
	\begin{align*}
		2\varepsilon_{\rm bias} \leq \sum_{k=1}^{\alpha n/t}\alpha^k\frac{\binom{ n/t}{k}}{\sqrt{\binom{n}{kt}}}(r-1)^{k/2} \sqrt{\sum_{\substack{S\in\Z_r^n: |S| = kt}}\|\widehat{\rho}(S)\|_1^2},
	\end{align*}
	where we also used that $\binom{\alpha n/t}{k} \leq \alpha^k \binom{n/t}{k}$ for $\alpha\in[0,1]$. In order to compute the above sum, we shall split it into two parts: one in the range $1\leq k < 4rm$, and the other in the range $4rm \leq k \leq \alpha n/t$. 

	\textbf{Sum I} ($1\leq k < 4rm$): in order to upper bound each term, pick $\delta = k/(4rm)$ in \cref{lem:hypercontractive}, so
	\begin{align*}
		\sum_{\substack{S\in\Z_r^n: |S| = kt}} \|\widehat{\rho}(S)\|_1^2 \leq \frac{1}{\delta^{kt}}\sum_{S\in\Z_r^n}\delta^{|S|} \|\widehat{f}(S)\|_1^2 \leq \frac{1}{\delta^{kt}}2^{2r\delta m} = \left(\frac{2^{1/(2t)}4rm}{k}\right)^{kt}.
	\end{align*}
	By using that $m \leq \frac{\gamma}{r^{1+1/t}} (\frac{\varepsilon^2}{\alpha})^{2/t} (n/t)^{1-2/t}$ and $\binom{q}{s}^2\binom{\ell q}{\ell s}^{-1} \leq (\frac{s}{q})^{(\ell-2)s}$ (see~\cite[Appendix~A.5]{shi2012limits}) for $q=n/t,s=k,\ell=t$, we thus have
	\begin{align*}
		\displaybreak[0]\sum_{k=1}^{4rm-1}\alpha^k\frac{\binom{n/t}{k}}{\sqrt{\binom{n}{kt}}}(r-1)^{k/2}\sqrt{\sum_{\substack{S\in\Z_r^n: \displaybreak[0]|S| = kt}} \|\widehat{\rho}(S)\|_1^2} &\leq \sum_{k=1}^{4rm-1} \alpha^k (r-1)^{k/2}\left(\frac{kt}{n}\right)^{(1-2/t)kt/2}\left(\frac{2^{1/(2t)}4rm}{k}\right)^{kt/2}\\\displaybreak[0]
		&\leq\sum_{k=1}^{4rm-1} \alpha^k (r-1)^{k/2}\left(\frac{2^{1/(2t)}4\gamma \varepsilon^{4/t}}{\alpha^{2/t}r^{1/t}k^{2/t}}\right)^{kt/2}\\\displaybreak[0]
		&\leq \sum_{k=1}^{4rm-1}\left(\frac{2^{1/4}(4\gamma)^{t/2} \varepsilon^2 }{k}\right)^{k} \leq \varepsilon^2
	\end{align*}
	for sufficiently small $\gamma$. 
	
	\textbf{Sum II} ($4rm \leq k \leq \alpha n/t$): first we note that the function $g(k) \triangleq \alpha^k (r-1)^{k/2}\binom{n/t}{k}/\sqrt{\binom{n}{kt}}$ is non-increasing in the interval $1\leq k \leq \alpha n/t \leq n/(2t)$. That is because $\alpha \sqrt{r-1} \leq 1$, and so
	\begin{align*}
		\displaybreak[0]\frac{g(k-1)}{g(k)} \geq \frac{\binom{n/t}{k-1}}{\sqrt{\binom{n}{kt-t}}}\frac{\sqrt{\binom{n}{kt}}}{\binom{n/t}{k}} = \sqrt{\frac{kt}{n-kt+t}\prod_{j=1}^{t-1}\frac{n-kt+j}{kt-j}} &\geq \sqrt{\frac{kt}{n-kt+t}\prod_{j=1}^{t-1}\frac{n-kt+j+1}{kt-j+1}}\displaybreak[0]\\
		&= \sqrt{\prod_{j=1}^{t-2}\frac{n-kt+j+1}{kt-j}} \geq 1,
	\end{align*}
	where we used that $\frac{a}{b} \geq \frac{a+s}{b+s}$ for all $a,b,s>0$ with $a\geq b$. Hence, and with the aid once more of \cref{lem:hypercontractive} with $\delta=1$ and the inequality $\binom{q}{s}^2\binom{\ell q}{\ell s}^{-1} \leq (\frac{s}{q})^{(\ell-2)s}$ (for $q=n/t,s=2m,\ell=t$) in order to bound $g(4rm)$,
	\begin{align*}
		\sum_{k=4rm}^{\alpha n/t}\alpha^k\frac{\binom{n/t}{k}}{\sqrt{\binom{n}{kt}}}(r-1)^{k/2}\sqrt{\sum_{\substack{S\in\Z_r^n: |S| = kt}} \|\widehat{\rho}(S)\|_1^2} &\leq g(4rm) \sum_{k=4rm}^{\alpha n/t}\sqrt{\sum_{\substack{S\in\Z_r^n: |S| = kt}} \|\widehat{\rho}(S)\|_1^2}\\
		&\leq g(4rm)\sqrt{\frac{\alpha n}{t}}\sqrt{\sum_{S\in\Z_r^n} \|\widehat{\rho}(S)\|_1^2}\tag{Cauchy-Schwarz}\\
		&\leq \left(\alpha \sqrt{r-1}\right)^{4rm}\left(\frac{4rm}{n/t}\right)^{2(t-2)rm}\sqrt{\frac{\alpha n}{t}}2^{(r-1)m}\\
		&\leq \left(2^{1/4}\alpha \sqrt{r-1}\right)^{4rm} \left(\frac{(4\gamma)^{t/2}\varepsilon^2}{\alpha \sqrt{r}(n/t)}\right)^{4(1-2/t)rm}\sqrt{\frac{\alpha n}{t}}\\
		&\leq \varepsilon^2,
	\end{align*}
	where in the last step we used that $m\geq 1 \implies 4(1-2/t)m \geq 1$ (so $n$ is in the denominator and $\varepsilon^{4(1-2/t)m} \leq \varepsilon$) and picked $\gamma$ sufficiently small. 
	
    Finally, merging both results, we get that, if $m \leq \frac{\gamma}{r^{1+1/t}}(\frac{ \varepsilon^2}{\alpha})^{2/t} (n/t)^{1-2/t}$, then $\varepsilon_{\rm bias} \leq \varepsilon^2$.
\end{proof}

A very similar classical communication lower bound for the $r\text{-}\HH(\alpha,t,n)$ problem can be proven.
\begin{theorem}
    \label{thr:thr_classicalhh}
    Any one-way classical protocol that achieves advantage $\varepsilon>0$ for the $r\text{-}\HH(\alpha,t,n)$ problem with $t\geq 2$ and $\alpha \leq 1/2$ requires $\Omega(r^{-1}(\varepsilon^4/\alpha)^{1/t}(n/t)^{1-1/t})$ bits of communication.
\end{theorem}
The proof is very similar to that of past works~\cite{gavinsky2007exponential,verbin2011streaming,DBLP:conf/approx/GuruswamiT19} and we include it in \cref{app:appB} for completeness. We now conclude this section with a remark that improves the $r$ dependence of the $\alpha$ parameter.
\begin{remark}
    \label{rem:alphadependence}
    The dependence of $\alpha$ on $r$ can be improved. For example, we can improve the bound in {\rm \cref{eq:alphadependence}} by observing that $\binom{k}{k_1,\dots,k_{r-1}}^2\binom{kt}{k_1t,\dots,k_{r-1}t}^{-1} \leq 1$, which can be seen from the identity $\binom{k}{k_1,\dots,k_{r-1}} = \binom{k_1}{k_1}\binom{k_1+k_2}{k_2}\cdots\binom{k_1+k_2+\cdots+k_{r-1}}{k_{r-1}}$ and the inequality $\binom{q}{s}^2\binom{\ell q}{\ell s}^{-1} \leq (\frac{s}{q})^{(\ell-2)s} \leq 1$. Hence
    \begin{align*}
        \sum_{\substack{k_1,\dots,k_{r-1}\geq 0 \\ \sum_{j=1}^{r-1} k_j = k}} \frac{k!^2}{(kt)!}\prod_{j=1}^{r-1}\frac{(k_jt)!}{k_j!^2} &= \sum_{\substack{k_1,\dots,k_{r-1}\geq 0 \\ \sum_{j=1}^{r-1} k_j = k}} \frac{\binom{k}{k_1,\dots,k_{r-1}}^2}{\binom{kt}{k_1t,\dots,k_{r-1}t}} \leq \sum_{\substack{k_1,\dots,k_{r-1}\geq 0 \\ \sum_{j=1}^{r-1} k_j = k}} 1 = \binom{k+r-2}{k},
    \end{align*}
    which is better than $(r-1)^{k}$. By bounding
    \begin{align*}
        \binom{k+r-2}{k} \leq e^k\left(1 + \frac{r-2}{k}\right)^{k},
    \end{align*}
    the new function $g(k) \triangleq \alpha^k \sqrt{\binom{k+r-2}{k}}\binom{n/t}{k}/\sqrt{\binom{n}{kt}}$ is still non-increasing in the interval $4rm \leq k \leq \alpha n/t \leq n/(2t)$ if now
    \begin{align*}
        \alpha \leq e^{-1/2}\operatorname*{\min}_{4rm\leq k \leq \alpha n/t}\sqrt{\frac{k}{k+r-2}} = e^{-1/2}\sqrt{\frac{4rm}{4rm+r-2}}.
    \end{align*}
    For $m\gg 1$, $\alpha$ is essentially independent of $r$, and hence $\alpha \leq \min(1/2,e^{-1/2}) = 1/2$.
\end{remark}

\subsection{Quantum streaming lower bound for Unique Games on hypergraphs}

The Unique Games problem is a generalization of the classical Max-Cut and can in fact be viewed as constraint satisfaction problems on a graph but over a larger alphabet. Consider a graph on $n$ vertices $x_1,\ldots,x_n$ and edges in $E$. The constraint on an arbitrary edge $(i,j)\in E$ is specified by a permutation $\pi_{i,j}:\Z_r\rightarrow \Z_r$ and the goal is to find an assignment of $x_1,\ldots,x_n\in \Z_r$ that maximizes 
$$
\sum_{(i,j)\in E} [\pi_{i,j}(x_i)=x_j].
$$ 
In this section, we consider a generalization of Unique Games to hypergraphs.
\begin{definition}[Unique Games instance on hypergraphs]
\label{def:UGhyper}
    A hypergraph $H=(V,E)$ is defined on a vertex set $V$ of size $n$ with $t$-sized hyperedges $E$ (i.e., $t$-sized subsets of $V$). Given a linear constraint on a hyperedge $e\in E$, i.e., a linear function $\pi_e:\Z_r^t\rightarrow \Z_r$, the goal is to compute
    $$
        \max_{x\in \Z_r^n}\sum_{e\in E} [\pi_{e}(x_e) = 0],
    $$
    where $x_e$ corresponds to the set of vertex-assignment in the hyperedge $e\in E$.    
\end{definition}
\begin{definition}
    Let $H=(V,E)$ be a hypergraph and let $\operatorname{OPT}$ be the optimal value of the Unique Games on $H$. {For $\gamma \geq 1$}, a randomized algorithm gives a $\gamma$-approximation to a Unique Games instance with failure probability $\delta\in[0,1/2)$ if, on any input hypergraph $H$, it outputs a value in the interval $[\operatorname{OPT}/\gamma, \operatorname{OPT}]$ with probability at least $1-\delta$.
\end{definition}
A uniformly random assignment of $x\in \Z_r^n$ to the vertex set $V$ will satisfy a $1/r$-fraction of the hyperedges, since each linear constraint $\pi_e(x_e)$ is satisfied with probability $1/r$. This gives a trivial $r$-approximation algorithm for the problem above. Below we show that any better than trivial approximation requires space that scales as $n^\beta$ for constant $\beta>0$.
\begin{theorem}
\label{thm:streaminguglowerbound}
    Let $r,t\geq 2$ be integers. Every {quantum} streaming algorithm giving an $(r-\varepsilon)$-approximation for Unique Games on hypergraphs (as in {\rm \cref{def:UGhyper}}) with at most $t$-sized hyperedges with alphabet size $r$ and success probability at least $2/3$ over its internal randomness, needs $\Omega((n/t)^{1-2/t})$ space (which hides dependence on $r,\varepsilon$). 
\end{theorem}
 The proof of this theorem combines techniques used by Guruswami and Tao~\cite{DBLP:conf/approx/GuruswamiT19} and Kapralov, Khanna, and Sudan~\cite{kapralov2014streaming}. Akin to these works, based on the Hidden Matching problem, we will construct instances of the hypergraph for which a Unique Games instance is hard to solve space-efficiently in the streaming model. 

\paragraph{Input distributions.} To this end, we construct two distributions $\Y$ and $\N$ such that $\Y$ is supported on satisfiable Unique Games instances and $\N$ is supported on instances for which at most an $O(1/r)$-fraction of the constraints is satisfied. We now define these instances in a multi-stage way (using $k$ stages). First, sample $k$ independent $\alpha$-partial $t$-hypermatchings on $n$ vertices and then construct a hypergraph $G$ by putting together all the hyperedges from these $k$ stages. Note that $G$ still has $n$ vertices, while the number of hyperedges is $k\cdot \alpha n /t$ (since each stage has $\alpha n/t$ many hyperedges and we allow multiple hyperedges should they be sampled). Now we specify the constraints $\pi_e$ in \cref{def:UGhyper} for the $\Y,\N$ distributions:
\begin{itemize}
    \item $\Y$ distribution: sample $z\in \Z_r^n$ and for each $e\in E$, let {$\pi_e(x_e)=\sum_{i\in e}(x_i-z_i)$} (where by $i\in e$ we mean all the vertices in the hyperedge $e$).
        \item $\N$ distribution: for each $e\in E$, pick a uniform $q\in \Z_r$ and let {$\pi_e(x_e)=q - \sum_{i\in e}x_i$.}
\end{itemize}
It is clear that, in the $\Y$ distribution, the optimal solution is when all the $x_1,\ldots,x_n$ are just set to $z_1,\ldots,z_n$. Below we show that for the $\N$ distribution, the value of the optimal solution is at most $(1+\varepsilon)/r$ with high probability.

\begin{lemma}
    \label{lem:NOdistribution}
    Let $\varepsilon\in (0,1)$. If $k=O(rt\log(r)/(\alpha\varepsilon^2))$, then for the Unique Games instance sampled from {the distribution $\N$} above, the optimal fraction of satisfiable constraints  (i.e., number of hyperedges $e\in E$ for which $\pi_e(\cdot)$ evaluates to {$0$}) over all possible vertex labelling is at most $(1+\varepsilon)/r$ with high probability.
\end{lemma}
\begin{proof}
    The proof of this lemma is similar to the proof in~\cite[Lemma~4.1]{DBLP:conf/approx/GuruswamiT19}. Fix an assignment $x\in \Z_r^n$. Let $X^\ell_{e}$ be the random variable that indicates that the hyperedge $e\in E$ appears in the $\ell$-th stage and is satisfied by $x$. Let $S=\sum_{\ell,e}X^{\ell}_e$. The expectation of $S$ is $k\alpha n/t \cdot 1/r$, since the total number of hyperedges is $\alpha n/t$ for each of the $k$ stages and the probability that a uniform $x$ satisfies a $t$-hyperedge (i.e., probability that $\sum_{e\in E} x_e=q$ for some fixed $q$) is $1/r$. Using the same analysis as in~\cite{DBLP:conf/approx/GuruswamiT19}, we can show that the variables $X^\ell_{e}$ are negatively correlated. Indeed, first note that hyperedges from different stages are independent. Now suppose we know that the random variables $X^{\ell}_{e_1},\dots,X^{\ell}_{e_s}$ have value $1$, and we also know a hyperedge $e\in E$. If $e\cap e_u \neq \emptyset$ for some $u\in[s]$, then $X^\ell_e=0$, since the hyperedges of a given stage form a matching. Otherwise, the conditional expectation of $X^\ell_e$ (conditioned on $e\cap e_u = \emptyset$ for all $u\in[s]$) is $\frac{\alpha n/t - s}{r}\binom{n-ts}{t}^{-1}$, which is less than its unconditional expectation of $\frac{\alpha n/t}{r}\binom{n}{t}^{-1}$. Therefore, in all cases one has $\mathbb{E}[X^\ell_e|X^\ell_{e_1}=\cdots=X^\ell_{e_s} = 1] \leq \mathbb{E}[X^\ell_e]$, which means negative correlation.
    
    Hence, a Chernoff bound for negative-correlated variables {(see, e.g.,~\cite[Theorem~3.2]{Panconesi1997randomized})} yields
    $$
        \Pr[S\geq (1+\varepsilon)(k\alpha n/t)/r]\leq \exp(-\varepsilon^2 k\alpha n /(3rt))=\exp(-O(n\log r)),
    $$
    where the inequality used the choice of $k$. Applying a union bound over the set of $x\in \Z_r^n$ concludes the proof of the lemma.
\end{proof}

\paragraph{Reduction to Hypermatching.} The reduction to $r$-ary Hidden Hypermatching is similar to the analysis used by Guruswami and Tao~\cite{DBLP:conf/approx/GuruswamiT19}, but now it is from quantum streaming algorithms to one-way quantum communication complexity. The main lemma that we need is the following. 
\begin{lemma}
\label{lem:reductionfromyesnotobhm}
    Let $\varepsilon>0$. If there is a streaming algorithm using at most $c$ qubits of space that distinguishes between the $\Y$ and $\N$ distributions on Unique Games instances (with $k$ stages) with bias $1/3$, then there is a $c$-qubit protocol that distinguish between the $\mathsf{YES}$ and $\mathsf{NO}$ distributions of $r\text{-}\HH(\alpha,t,n)$ with bias $\Omega(1/k)$. 
\end{lemma}

In order to prove this lemma we need a few definitions and facts. First, towards proving the lemma above, let us assume there is a $c$-qubit streaming $\A$ for \cref{lem:reductionfromyesnotobhm}. During the execution of the streaming protocol on instances from the $\Y$ and $\N$ distributions, let the memory content after receiving the $i$th stage constraints be given by the $c$-qubit quantum states $\ket{\phi_i^\Y}$ and $\ket{\phi_i^\N}$, respectively.\footnote{Without loss of generality, we assume they are pure states---this only affects the cost of the protocol by a constant factor (since one can always purify mixed quantum states by doubling the dimension).} Assume that $|\phi_0^\Y\rangle = |\phi_0^\N\rangle = 0$. Using the notion of informative index from~\cite[Definition~6.2]{kapralov2014streaming}, we say an index $j\in \{0,\ldots,k-1\}$ is $\delta$-informative if 
$$
    \big\|\ket{\phi^\Y_{j+1}}-\ket{\phi^\N_{j+1}}\big\|_1\geq \big\|\ket{\phi^\Y_{j}}-\ket{\phi^\N_{j}}\big\|_1+\delta.
$$
With this definition it is not hard to see the following fact, which follows from a simple triangle inequality.
\begin{fact}
    \label{fact:informativeindex}
    Suppose there exists a streaming protocol for distinguishing the $\Y,\N$ distributions with advantage $\geq 1/3$, then there exists a $\Omega(1/k)$-informative index.
\end{fact}
Suppose $j^*$ is an $\Omega(1/k)$-informative index for the streaming protocol $\A$. Using this we devise a communication protocol for $r\text{-}\HH(\alpha,t,n)$ with bias $\Omega(1/k)$ as follows: suppose Alice has a string $x\in \Z_r^n$ and Bob has $w\in \Z_r^{\alpha n/t}$ and a hypermatching $M\in\mathcal{M}_{t,n}^\alpha$.
\begin{enumerate}
    \item Alice samples $j^*$ many $\alpha$-partial $t$-hypermatchings and runs the streaming algorithm $\A$ on Unique Games constraints for the first $j^*$ stages that follow the $\Y$ distribution with $z=x$. She then sends the memory contents after these $j^\ast$ stages to Bob.
    \item Bob assigns the constraints $\sum_{i\in e}x_i=w_e$, where $e\in M$, according to his inputs $w,M$. He then continues running $\A$ on these constraints as the $(j^*+1)$th~stage. 

    Let $\ket{s}$ be the quantum state that Bob gets after running $\mathcal{A}$.
    \item Let $\ket{\phi^\YES}$ and $\ket{\phi^\NO}$ be the resulting quantum states under the two cases, depending on $w$'s distribution (these can be computed by Bob since $\mathcal{A}$ is known). Bob can distinguish between $\ket{\phi^\YES}$ and $\ket{\phi^\NO}$ with bias $\frac{1}{2}\big\||\phi^\YES\rangle-|\phi^\NO\rangle\big\|_1$ by measuring the state $|s\rangle$ with a suitable POVM, according to \cref{lem:lem3.5.c3}.
\end{enumerate}

We are now ready to prove \cref{lem:reductionfromyesnotobhm}.
\begin{proof}[Proof of {\rm \cref{lem:reductionfromyesnotobhm}}]
    We argue that the above protocol achieves a $\Omega(1/k)$ bias in distinguishing between the $\mathsf{YES}$ and $\mathsf{NO}$ distributions from $r\text{-}\HH(\alpha,t,n)$. To this end, let $U$ be the unitary that maps the quantum state after stage $j^*$ and constraints of stage $j^*+1$ (which is classical) to the quantum state after $j^*+1$. Thus we have $\ket{\phi^\YES}=\ket{\phi^\Y_{j^*+1}}=U\ket{\phi^\Y_{j^*},C^\Y}$ and $\ket{\phi^\NO}=U\ket{\phi^\Y_{j^*},C^\N}$, where $C^\Y$ and $C^\N$ are the constraints corresponding to the $\YES$ and $\NO$ distributions, respectively, and, similarly, we have $\ket{\phi^\N_{j^*+1}}=U\ket{{\phi^\N_{j^*}},C^\N}$. Then, we have
    \begin{align*}
        \big\|\ket{\phi^\YES}-\ket{\phi^\NO}\big\|_1 &\geq      \big\|\ket{\phi^\Y_{j^*+1}}-\ket{\phi^\N_{j^*+1}}\big\|_1-     \big\|\ket{\phi^\NO}-\ket{\phi^\N_{j^*+1}}\big\|_1\\
        &\geq \big\|\ket{\phi^\Y_{j^*+1}}-\ket{\phi^\N_{j^*+1}}\big\|_1 -     \big\|\ket{\phi^\Y_{j^*}}-\ket{\phi^\N_{j^*}}\big\|_1 =\Omega(1/k),
    \end{align*}
    where the second inequality used that  $\big\|\ket{\phi^\NO}-\ket{\phi^\N_{j^*+1}}\big\|_1 = \big\|U\ket{\phi^\Y_{j^*},C^\N}-U\ket{\phi^\N_{j^*},C^\N}\big\|_1\leq \|\ket{\phi^\Y_{j^*}}-\ket{\phi^\N_{j^*}}\|_1$ (since unitaries preserve norms) and the third inequality is because $j^*$ is an informative index. Hence in Step (3) of the procedure above, the bias of Bob in obtaining the right outcome is~$\Omega(1/k)$.
\end{proof}

\begin{proof}[Proof of {\rm \cref{thm:streaminguglowerbound}}] Finally, by picking $k=O(rt\log(r)/(\alpha\varepsilon^2))$ in order to invoke \cref{lem:NOdistribution} and using our lower bound in \cref{thm:hypermatchinglowerbound} with $\alpha= O(1)$, we get our desired lower bound of
\[
\Omega(r^{-(1+1/t)}(k^2\alpha)^{-2/t}(n/t)^{1-2/t})=\Omega((n/t)^{1-2/t}). \qedhere
\]
\end{proof}
\begin{remark}
    It is possible to prove a classical version of {\rm \cref{thm:streaminguglowerbound}} by using the classical lower bound on $r$-$\mathsf{HH}(\alpha,t,n)$, which leads to $\Omega((n/t)^{1-1/t})$ classical space (hiding dependence on $r,\varepsilon$). Such bound, though, is already subsumed by $\Omega(n)$ from Chou et al.~{\rm \cite{chou2022linear}}.
\end{remark}

\section{Locally Decodable Codes}
In this section we prove our lower bound on locally decodable codes over $\Z_r$. Before that, let us first formally define an $\mathsf{LDC}$.
\begin{definition}[Locally decodable code]
    A $(q,\delta,\varepsilon)$-locally decodable code over $\Z_r$ is a function $C:\Z_r^n\rightarrow \Z_r^N$ that satisfies the following: for every $x\in \Z_r^n$ and $i\in [n]$, there exists a (randomized) algorithm $\A$ that, on any input $y\in \Z_r^N$ that satisfies $d(y,C(x))\leq \delta N$, makes $q$ queries to $y$ non-adaptively and outputs a number $\A^{y}(i)\in \Z_r$ that satisfies $\Pr[\A^{y}(i)=x_i]\geq 1/r+\varepsilon$ (where the probability is only taken over the randomness of $\A$).
\end{definition}
As is often the case when proving $\mathsf{LDC}$ lower bounds, we use the useful fact proven by Katz and Trevisan~\cite{katz2000efficiency} that, without loss of generality, one can assume that an $\mathsf{LDC}$ is smooth, i.e., the queries made by $\A$ have ``reasonable" probability over all indices, and that $\A$ makes queries to a codeword (and not a corrupted codeword). We first formally define a smooth code below. 
\begin{definition}[Smooth code]
    We say $C:\Z_r^n\rightarrow \Z_r^N$ is a $(q,c,\varepsilon)$-smooth code if there exists a decoding algorithm $\A$ that satisfies the following: for every $x\in \Z_r^n$ and $i\in [n]$, $\A$ makes at most $q$ non-adaptive queries to $C(x)$ and outputs $\A^{C(x)}(i)\in \Z_r$ such that $\Pr[\A^{C(x)}(i)=x_i]\geq 1/r+\varepsilon$ (where the probability is only taken over the randomness of $\A$). Moreover, for every $x\in \Z_r^n$, $i\in [n]$ and $j\in [N]$, on input $i$, the probability that $\A$ queries the index $j$ in $C(x)\in \Z_r^N$ is at most $c/N$.
\end{definition}
Crucially note that smooth codes only require a decoder to recover $x_i$ when given access to an actual codeword, unlike the standard definition of $\mathsf{LDC}$ where a decoder is given a noisy codeword. With this definition in hand, we state a theorem of Katz and Trevisan.
\begin{theorem}[\cite{katz2000efficiency}]
    \label{thr:ldc-smoothcode-equivalence}
    A $(q,\delta,\varepsilon)$-$\mathsf{LDC}$ $C:\Z_r^n\rightarrow \Z_r^N$ is a $(q,q/\delta,\varepsilon)$-smooth code.
\end{theorem}
\noindent We remark that a converse to this theorem holds: a $(q, c, \varepsilon)$-smooth code is a $(q,\delta,\varepsilon - c\delta)$-$\mathsf{LDC}$, since the probability that the decoder queries one of $\delta N$ corrupted positions is at most $(c/N)(\delta N) = c\delta$.

We now present our lower bound for $\mathsf{LDC}$s over $\mathbb{Z}_r$ using the non-commutative Khintchine's inequality. We thank Jop Bri\"{e}t for the proof.
\begin{remark}
    It is possible to prove a weaker lower bound $N= 2^{\Omega(\delta^2 \varepsilon^4 n/r^4)}$ using our matrix-valued hypercontractivity via the proof technique in~{\rm \cite[Theorem~11]{ben2008hypercontractive}}.
\end{remark}
\begin{theorem}
    \label{thm:2querylowerboundhyper}
    If $C:\Z_r^n\rightarrow \Z_r^N$ is a $(2,\delta,\varepsilon)$-$\mathsf{LDC}$, then $N= 2^{\Omega(\delta^2 \varepsilon^2 n/r^2)}$.
\end{theorem}
\begin{proof}
    We know that $C$ is also a $(2,2/\delta,\varepsilon)$-smooth code. Fix some $i\in[n]$. In order to decode $x_i$, we can assume, without loss of generality, that the decoder $\A$ picks some set $\{u,v\}$, where $u,v\in[N]$, with probability $p(u,v)$, queries those bits, and then outputs $f^{u,v}_i(C(x)_{u},C(x)_v)\in\mathbb{Z}_r$ that depends on the query-outputs. Given the smooth code property of outputting $x_i$ with probability at least $\frac{1}{r} + \varepsilon$ for every $x$, we have
    \begin{align*}
		\frac{1}{r} + \varepsilon &\leq \operatorname*{Pr}_{u,v\sim[N]}[f^{u,v}_i(C(x)_u,C(x)_v) = x_i] = \frac{1}{r}\sum_{k=0}^{r-1}\sum_{u,v\in[N]}p(u,v)\omega_r^{k(f^{u,v}_i(C(x)_u,C(x)_v) - x_i)} \iff \\
		r\varepsilon &\leq \sum_{k=1}^{r-1}\sum_{u,v\in[N]}p(u,v)\omega_r^{k(f^{u,v}_i(C(x)_u,C(x)_v) - x_i)},
    \end{align*}
    {where the if and only if comes from separating the term $k=0$, which equals $\sum_{u,v\in[N]}p(u,v) = 1$.}
 
	For $k\in[r-1]$, define the function $h_{i;k}^{u,v}:\mathbb{Z}_r^{2}\to\mathbb{C}$ by $h_{i;k}^{u,v}(x) = \omega_r^{kf_i^{u,v}(x)}$. Consider its Fourier transform, $\widehat{h_{i;k}^{u,v}}: \mathbb{Z}_r^{2}\to\mathbb{C}$, defined by
	\begin{align*}
		\widehat{h_{i;k}^{u,v}}(S) = \frac{1}{r^{2}}\sum_{x\in\mathbb{Z}_r^{2}}h_{i;k}^{u,v}(x)\omega_r^{-S\cdot x}.
	\end{align*}
	Hence we can write
	\begin{align*}
		\omega_r^{kf_i^{u,v}(C(x)_u,C(x)_v)} = \sum_{a,b\in\mathbb{Z}_r}\widehat{h_{i;k}^{u,v}}(a,b)\omega_r^{aC(x)_u + bC(x)_v}.
	\end{align*}
	Let $F^i_{k}\in\mathbb{C}^{rN\times rN}$ be the matrix defined as $(F^i_{k})_{(u,a),(v,b)} = p(u,v)\widehat{h_{i;k}^{u,v}}(a,b)$ and let $C(x)\in\mathbb{C}^{rN}$ be the vector defined as $C(x)_{(u,a)} = \omega_r^{aC(x)_u}$. Then
	\begin{align*}
		\sum_{u,v\in[N]} p(u,v)\omega_r^{kf_i^{u,v}(C(x)_u,C(x)_v)} = C(x)^{\top}F^i_{k}C(x),
	\end{align*}
	which means that
	\begin{align*}
		\sum_{k=1}^{r-1}\operatorname*{\mathbb{E}}_{x\sim\mathbb{Z}_r^n}\left[\omega_r^{-kx_i}C(x)^{\top}F^i_{k}C(x) \right] = \sum_{k=1}^{r-1}\sum_{u,v\in[N]} p(u,v)\operatorname*{\mathbb{E}}_{x\sim\mathbb{Z}_r^n}\left[\omega_r^{k(f_i^{u,v}(C(x)_u,C(x)_v) - x_i)}\right] \geq r\varepsilon.
	\end{align*}
	By summing over all $i\in[n]$, we finally get to
	\begin{align*}
		\sum_{k=1}^{r-1}\operatorname*{\mathbb{E}}_{x\sim\mathbb{Z}_r^n}\left[C(x)^{\top}\left(\sum_{i=1}^n \omega_r^{-kx_i}F^i_{k}\right)C(x) \right] \geq r\varepsilon n.
	\end{align*}
	The left-hand side can be upper bounded as follows:
	\begin{align*}
		\sum_{k=1}^{r-1}\operatorname*{\mathbb{E}}_{x\sim\mathbb{Z}_r^n}\left[C(x)^{\top}\left(\sum_{i=1}^n \omega_r^{-kx_i}F^i_{k}\right)C(x) \right] &\leq  \sum_{k=1}^{r-1}\operatorname*{\mathbb{E}}_{x\sim\mathbb{Z}_r^n}\left[\left\|\sum_{i=1}^n \omega_r^{-kx_i}F^i_{k}\right\|\|C(x)\|^2_2\right]\\
		&= rN \sum_{k=1}^{r-1}\operatorname*{\mathbb{E}}_{x\sim\mathbb{Z}_r^n}\left[\left\|\sum_{i=1}^n \omega_r^{-kx_i}F^i_{k}\right\|\right]\\
		&\leq 2rN\sqrt{2\log(2rN)} \sum_{k=1}^{r-1}\sqrt{\sum_{i=1}^n \|F_k^i\|^2},
	\end{align*}
	where we used \cref{lem:lem1} in the last step. Consider the submatrix $[p(u,v)\widehat{h_{i;k}^{u,v}}(a,b)]_{a,b}$ of $F_k^i$ for fixed $u,v\in[N]$. Since
	\begin{align*}
		\|[p(u,v)\widehat{h_{i;k}^{u,v}}(a,b)]_{a,b}\|^2 \leq \|[p(u,v)\widehat{h_{i;k}^{u,v}}(a,b)]_{a,b}\|_F^2 = \sum_{a,b\in\mathbb{Z}_r}p(u,v)^2 |\widehat{h_{i;k}^{u,v}}(a,b)|^2 = p(u,v)^2 \leq \frac{4}{\delta^2 N^2},
	\end{align*}
	where we used Parseval's identity and $p(u,v) \leq 2/\delta N$ for all $u,v\in[N]$ by the definition of smooth code, then $\|F_k^i\|^2 \leq \frac{4}{\delta^2 N^2}$. This finally leads to
	\[
		2rN\sqrt{2\log(2rN)}(r-1)\frac{2\sqrt{n}}{\delta N} \geq r\varepsilon n \implies \log(2rN) \geq \frac{\delta^2\varepsilon^2}{8r^2}n \implies N\geq \frac{1}{2r}2^{\delta^2\varepsilon^2 n/8r^2}. \qedhere
	\]
\end{proof}

\section{2-server private information retrieval}
\label{sec:PIR}

As mentioned in the introduction, the connection between $\mathsf{LDC}$s and \textsf{PIR} is well known since the results of~\cite{katz2000efficiency,goldreich2002lower}. In general, upper bounds on $\mathsf{LDC}$s  are derived via \textsf{PIR} schemes, which in turn means that our $\mathsf{LDC}$ lower bounds translate to $\textsf{PIR}$ lower bounds, which we illustrate below. We first define the notion of private information retrieval.
\begin{definition} 
    A one-round, $(1-\delta)$-secure, $k$-server private information retrieval $\mathsf{(PIR)}$ scheme with alphabet $\Z_r$, recovery probability $1/r+\varepsilon$, query size $t$ and answer size $a$, consists of a randomized user and $k$ deterministic algorithms $S_1,\ldots,S_k$ (the servers) that satisfy the following:
    \begin{enumerate}
        \item On input $i\in [n]$, the user produces $k$ queries $q_1,\ldots,q_k\in \Z_r^t$ and sends them to the $k$ servers respectively. The servers reply back with a string $a_j=S_j(x,q_j)\in\Z_r^a$, and based on $a_1,\ldots,a_k$ and $i$, the user outputs $b\in \Z_r$.
        \item For every $x\in \Z_r^n$ and $i\in [n]$, the output $b$ of the user satisfies $\Pr[b=x_i]\geq 1/r+\varepsilon$.
        \item For every $x\in \Z_r^n$ and $j\in[k]$, the distributions over $q_j$ (over the user's randomness) are $\delta$-close for different $i\in[n]$.
    \end{enumerate}
    We crucially remark that for the lower bounds that we present below, the function $S_j$ could be an arbitrary (not necessarily linear) function over $x_1,\ldots,x_n \in \Z_r$. 
\end{definition}

Our \textsf{PIR} lower bound follows from a result of Goldreich et al.~\cite{goldreich2002lower} and from a generalization of~\cite[Lemma~2]{DBLP:journals/jcss/KerenidisW04}. We remark that Goldreich et al.\ state the lemma below only for $r=2$, but the exact same analysis carries over to $r>2$. In the following we shall assume $\delta = 0$. 
\begin{lemma}[{\cite[Lemma~5.1]{goldreich2002lower}}]
    \label{lem:goldreich}
    If there is a classical $2$-server $\mathsf{PIR}$ scheme with alphabet $\Z_r$, query size $t$, answer size $a$, and recovery probability $1/r+\varepsilon$, then there is a $(2, 3, \varepsilon)$-smooth code $C:\Z_r^n\rightarrow (\Z_r^a)^m$ with $m \leq 6r^t$.
\end{lemma}
\begin{lemma}
    \label{lem:dewolf_generalisation}
    Let $r\geq 2$ be prime. Let $C:\Z_r^n \to (\Z_r^a)^m$ be a $(2,c,\varepsilon)$-smooth code. Then there is a $(2,cr^a,2\varepsilon/r^{a+2})$-smooth code $C':\Z_r^n \to \Z_r^{mr^a}$ that is good on average, i.e., there is a decoder $\mathcal{A}$ such that, for all $i\in[n]$,
    \begin{align*}
        \operatorname*{\mathbb{E}}_{x\sim\Z_r^a}\left[\operatorname{Pr}[\mathcal{A}^{C'(x)}(i) = x_i]\right] \geq \frac{1}{r} + \frac{2\varepsilon  }{r^{a+2}}.
    \end{align*}
\end{lemma}
\begin{proof}
    We form a new code $C'$ by transforming each old string $C(x)_j\in\Z_r^a$ using the Hadamard code into $C'(x)_j\in\Z_r^{r^a}$, $C'(x) = (\langle C(x),y\rangle)_{y\in\Z_r^a}$. The total length of $C'$ is $mr^a$. The new decoder uses the same randomness as the old one. Let $f:(\Z_r^a)^2\to\Z_r$ be the output function of the old decoder. Fix the queries $j,k\in[m]$. We now describe the new decoding procedure.

    First, if for $j,k$ the function $f$ is such that $\operatorname{Pr}_{x\sim\Z_r^n}[f(C(x)_j,C(x)_k)=x_i] \leq \frac{1}{r}$, then the new decoder outputs a random value in $\Z_r$, in which case it is as good as the old one for an average $x$. Consider now the case $\operatorname{Pr}_{x\sim\Z_r^n}[f(C(x)_j,C(x)_k) = x_i] = \frac{1}{r} + \eta$ for some $\eta > 0$. Then
    \begin{align*}
		\operatorname*{Pr}_{x\sim\Z_r^n}[f(C(x)_j,C(x)_k) = x_i]  = \frac{1}{r} + \eta \iff 
		\sum_{\ell=1}^{r-1}\operatorname*{\mathbb{E}}_{x\sim\Z_r^n}\big[\omega_r^{\ell(f(C(x)_j,C(x)_k) - x_i)}\big] = r\eta.
    \end{align*}
    For $\ell\in[r-1]$, define the function $h_{\ell}:(\Z_r^a)^{2}\to\mathbb{C}$ by $h_{\ell}(a,b) = \omega_r^{\ell f(a,b)}$. Consider its Fourier transform $\widehat{h}_{\ell}: (\Z_r^a)^{2}\to\mathbb{C}$. Hence we can write
    \begin{align*}
        h_{\ell}(x) = \sum_{S,T\in\Z_r^{a}}\widehat{h}_{\ell}(S,T)\omega_r^{S\cdot a + T\cdot b}
    \end{align*}
    and then
    \begin{align*}
        r\eta &= \sum_{\ell=1}^{r-1}\sum_{S,T\in\mathbb{Z}_r^a}\widehat{h}_\ell(S,T)\operatorname*{\mathbb{E}}_{x\sim\mathbb{Z}_r^n}\big[\omega_r^{S\cdot C(x)_j + T\cdot C(x)_k - \ell x_i}\big] \\
        &\leq \sqrt{\sum_{\ell=1}^{r-1}\sum_{S,T\in\mathbb{Z}_r^a}|\widehat{h}_\ell(S,T)|^2} \sqrt{\sum_{\ell=1}^{r-1}\sum_{S,T\in\mathbb{Z}_r^a}\left|\operatorname*{\mathbb{E}}_{x\sim\mathbb{Z}_r^n}\big[\omega_r^{S\cdot C(x)_j + T\cdot C(x)_k - \ell x_i}\big]\right|^2}\\
        &= \sqrt{r-1} \sqrt{\sum_{\ell=1}^{r-1}\sum_{S,T\in\mathbb{Z}_r^a}\left|\operatorname*{\mathbb{E}}_{x\sim\mathbb{Z}_r^n}\big[\omega_r^{S\cdot C(x)_j + T\cdot C(x)_k - \ell x_i}\big]\right|^2},
    \end{align*}
    where we used Parseval's identity $\sum_{S,T\in\mathbb{Z}_r^a}|\widehat{h}_\ell(S,T)|^2 = 1$ for all $\ell\in[r-1]$. It then follows that there are $S_0,T_0\in\mathbb{Z}_r^a$ and $\ell_0\in[r-1]$ such that
    \begin{align*}
        \left|\operatorname*{\mathbb{E}}_{x\sim\mathbb{Z}_r^n}\big[\omega_r^{S_0\cdot C(x)_j + T_0\cdot C(x)_k - \ell_0 x_i}\big]\right| \geq \frac{\eta}{r^{a}}.
    \end{align*}
    We now use the following useful lemma whose proof is left to the end of the section.
    \begin{lemma}
        \label{lem:lem-exptoprob}
        Let $f:\Z_r^n\to\Z_r$ and $\alpha\in[0,1]$. Suppose that $\Big|{\operatorname*{\mathbb{E}}_{x\sim\Z_r^n}}\big[\omega_r^{f(x)}\big]\Big| \geq \alpha$. Then
        \begin{align*}
            \max_{y\in\Z_r}\operatorname*{Pr}_{x\sim\Z_r^n}[f(x) = y] \geq \frac{1}{r} + \frac{2\alpha}{r^2}.
        \end{align*}
    \end{lemma}
    According to the above lemma, there is $y\in\Z_r$ such that
    \begin{align*}
        \operatorname*{Pr}_{x\sim\Z_r^n}[S_0\cdot C(x)_j + T_0\cdot C(x)_k - y = \ell_0 x_i] \geq \frac{1}{r} + \frac{2\eta}{r^{a+2}}.
    \end{align*}
    Since both numbers $S_0\cdot C(x)_j$ and $T_0\cdot C(x)_k$ are in the code $C'$, it is possible to recover $x_i$ (the inverse of $\ell_0$ is well defined since $r$ is prime) with just two queries and average probability at least $1/r + 2\eta/r^{a+2}$. Averaging over the classical randomness, i.e., $j,k$ and $f$, gives the lemma. Finally, the $c$-smoothness of $C$ translates into $cr^a$-smoothness of $C'$.
\end{proof}
We now  get the following main theorem.
\begin{theorem}\label{thr:pir_lower_bound}
    Let $r\geq 2$ be prime. A classical $2$-server $\mathsf{PIR}$ scheme with alphabet $\Z_r$, query size $t$, answer size $a$, and recovery probability $1/r + \varepsilon$ satisfies $t\geq \Omega\big((\varepsilon^2 n/r^{4a+6} - a)/\log{r}\big)$. 
\end{theorem}
\begin{proof}
    By \cref{lem:goldreich}, there is a $(2,3,\varepsilon)$-smooth code $C:\Z_r^n\rightarrow (\Z_r^a)^m$ with $m\leq 6r^t$. \cref{lem:dewolf_generalisation} allows us to transform $C$ into a $(2,3r^a,2\varepsilon/r^{a+2})$-smooth code $C':\Z_r^n\to\Z_r^{mr^a}$. Using \cref{thm:2querylowerboundhyper} on $C'$ (note the theorem can be applied directly to a smooth code that is good on average) yields
    $$
        mr^a\geq 2^{\Omega(\varepsilon^2 n/r^{4a+6})},
    $$
    and since $m=O(r^t)$, we get the desired lower bound in the theorem statement.
\end{proof}
\begin{proof}[Proof of {\rm \cref{lem:lem-exptoprob}}]
    First notice that
    \begin{align*}
        \Big|{\operatorname*{\mathbb{E}}_{x\sim\Z_r^n}}\big[\omega_r^{f(x)}\big]\Big| = \left|\sum_{k\in\Z_r}{\operatorname*{Pr}_{x\sim\Z_r^n}}[f(x)=k]\omega_r^k  \right|,
    \end{align*}
    hence, for any $\beta\in\mathbb{R}$,
    \begin{align*}
        \sum_{k\in\Z_r}\left|{\operatorname*{Pr}_{x\sim\Z_r^n}}[f(x)=k] - \beta\right| &= \sum_{k\in\Z_r}\left|{\operatorname*{Pr}_{x\sim\Z_r^n}}[f(x)=k]\omega_r^k - \beta\omega_r^k\right|\\
        &\geq \left|\sum_{k\in\Z_r}\left(\operatorname*{Pr}_{x\sim\Z_r^n}[f(x)=k]\omega_r^k - \beta\omega_r^k\right)\right| = \Big|{\operatorname*{\mathbb{E}}_{x\sim\Z_r^n}}\big[\omega_r^{f(x)}\big]\Big| \geq \alpha.
    \end{align*}
    Therefore, there exists $y\in\Z_r$ such that
    \begin{align*}
        \left|{\operatorname*{Pr}_{x\sim\Z_r^n}}[f(x)=y] - \beta\right| \geq \frac{\alpha}{r}.
    \end{align*}
    Take $\beta = \frac{1}{r} - \alpha\frac{r-2}{r^2}$. Then either $\operatorname*{Pr}_{x\sim\Z_r^n}[f(x)=y] \geq \beta + \frac{\alpha}{r} = \frac{1}{r} + \frac{2\alpha}{r^2}$, and so the lemma follows, or $\operatorname*{Pr}_{x\sim\Z_r^n}[f(x)=y] \leq \beta - \frac{\alpha}{r} = \frac{1}{r} - 2\alpha\frac{r-1}{r^2}$. In the second case,
    \begin{align*}
        \sum_{k\in\Z_r}\operatorname*{Pr}_{x\sim\Z_r^n}[f(x)=k] = 1 \implies \sum_{k\neq y}\operatorname*{Pr}_{x\sim\Z_r^n}[f(x)=k] \geq 1 - \frac{1}{r} + 2\alpha\frac{r-1}{r^2},
    \end{align*}
    and so there exists $z\in\Z_r$ such that
    \[
         \operatorname*{Pr}_{x\sim\Z_r^n}[f(x)=z] \geq \frac{1}{r-1} - \frac{1}{r(r-1)} + \frac{2\alpha}{r^2} = \frac{1}{r} + \frac{2\alpha}{r^2}. \qedhere
    \]
\end{proof}

\begin{acks}
SA firstly thanks T.S.\ Jayram for introducing him to this problem (and several discussions thereafter) on proving quantum bounds for streaming algorithms while participating in the program ``Quantum Wave in Computing” held at Simons Institute for the Theory for Computing. Special thanks to Jop Bri\"et for several clarifications on hypercontractivity and for helping us improving our $\mathsf{LDC}$ lower bound using the non-commutative Khintchine's inequality. We also thank Ronald de Wolf for useful discussions regarding hypercontractivity and $\mathsf{LDC}$s, and Mario Szegedy for discussions during the initial stages of this project. We are also very thankful to Keith Ball and Eric Carlen for the help in understanding their proof of uniform convexity for trace ideals. We thank several referees for their comments on the manuscript.
\end{acks}

\DeclareRobustCommand{\DE}[2]{#2}
\bibliographystyle{ACM-Reference-Format}
\bibliography{stream}

\appendix

\section{Classical Hidden Hypermatching lower bound}
\label{app:appB}

The general idea behind the proof of \cref{thr:thr_classicalhh} is already well established and is a simple generalization of results from~\cite{gavinsky2007exponential,verbin2011streaming,DBLP:conf/approx/GuruswamiT19,doriguello2020exponential}. We shall need the following well-known fact and a generalization of the KKL inequality~\cite{kahn1989influence}.

\begin{fact}
	\label{fact:fact1}
	Given just one sample, the best success probability in distinguishing between two probability distributions $p$ and $q$ is $\frac{1}{2} + \frac{1}{4}\|p-q\|_{\operatorname{tvd}}$.
\end{fact}

\begin{lemma}[Generalized KKL inequality]
	\label{lem:kkl-largefields}
	Let $A\subseteq\Z_r^n$ and let $f:\Z_r^n\to\{0,1\}$ be its indicator function ($f(x) = 1$ iff $x\in A$). Then, for every $\delta\in[0,1/r]$,
	\begin{align*}
		\sum_{S\in\Z_r^n}\delta^{|S|}|\widehat{f}(S)|^2 \leq \left(\frac{|A|}{r^n}\right)^{2/(1+r\delta)}.
	\end{align*}
\end{lemma}
\begin{proof}
    Apply the hypercontractive inequality to real-valued functions with $p = 1 + r\delta$.
\end{proof}

\begin{theorem}\label{thr:classical_communication_lower_bound}
    Any classical protocol that achieves advantage $\varepsilon>0$ for the $r\text{-}\HH(\alpha,t,n)$ problem with $t\geq 2$ and $\alpha \leq 1/2$ needs $\Omega(r^{-1}(\varepsilon^4/\alpha)^{1/t}(n/t)^{1-1/t})$ bits of communication from Alice to Bob.
\end{theorem}
\begin{proof}
    By the minimax principle, it suffices to analyse \emph{deterministic} protocols under some `hard' input distribution.  For our input distribution, Alice and Bob receive $x\in\Z_r^n$ and $M\in\mathcal{M}_{t,n}^\alpha$, respectively, uniformly at random, while Bob's input $w\in\Z_r^{\alpha n/t}$ is drawn from the distribution $\mathcal{D} \triangleq \frac{1}{2}\mathcal{D}^{\YES} + \frac{1}{2}\mathcal{D}^{\NO}$, i.e., with probability $1/2$ is comes from $\mathcal{D}^{\YES}$, and with probability $1/2$ it comes from $\mathcal{D}^{\NO}$. 
    	    
	Fix a small constant $\varepsilon>0$ and let $c = \gamma r^{-1}(\varepsilon^4/\alpha)^{1/t}(n/t)^{1-1/t}$ for some universal constant $\gamma$. Consider any classical deterministic protocol that communicates at most $C \triangleq  c - \log(1/\varepsilon)$ bits. Such protocol partitions the set of all $r^n$ $x$'s into $2^C$ sets. These sets have size $r^n/2^C$ on average, and by a counting argument, with probability $1-\varepsilon$, the set $A$ corresponding to Alice's message has size at least $\varepsilon r^n/2^C = r^n/2^c$. Given Alice's message, Bob knows that the random variable $X$ corresponding to her input was drawn uniformly at random from $A$, and he also knows his input $M$. Therefore his knowledge of the random variable $MX$ is described by the distribution
	\begin{align*}
		p_M(z) \triangleq \operatorname{Pr}[MX=z|M,A] = \frac{|\{x\in A|Mx=z\}|}{|A|}.
	\end{align*}
	
	Given one sample of $w\in\Z_r^{\alpha n/t}$, Bob must decide whether it came from $\mathcal{D}^{\YES}$ (the distribution $MX$) or from $\mathcal{D}^{\NO}$ (the uniform distribution $U$). According to \cref{fact:fact1}, the advantage of any classical protocol in distinguishing between $p_M$ and $U$ is upper bounded by $\frac{1}{4}\|p_M - U\|_{\text{tvd}}$. We prove in \cref{thr:raryhidden2} below that, if $\alpha\leq 1/2$ and $c \leq \frac{\gamma}{r}(\frac{\varepsilon^4}{\alpha})^{1/t} (n/t)^{1-1/t}$, then the average advantage over all hypermatchings $M$ is at most $\varepsilon^2/4$, i.e.,
	\begin{align*}
		\operatorname*{\mathbb{E}}_{M\sim\mathcal{M}_{t,n}^\alpha}[\|p_M - U\|_{\text{tvd}}] \leq \varepsilon^2.
	\end{align*}
	Therefore, by Markov's inequality, for at least a $(1-\varepsilon)$-fraction of $M$, the advantage in distinguishing between $p_M$ and $U$ is $\varepsilon/4$-small. Hence, Bob's total advantage over randomly guessing the right distribution will be at most $\varepsilon$ (for the event that $A$ is too small) plus $\varepsilon$ (for the event that $M$ is such that the distance between $MX$ and $U$ is more than $\varepsilon$) plus $\varepsilon/4$ (for the advantage over random guessing when $\|p_M - U\|_{\text{tvd}} \leq \varepsilon$), and so $c=\Omega(r^{-1}(\varepsilon^4/\alpha)^{1/t}(n/t)^{1-1/t})$.
\end{proof}

\begin{theorem}
	\label{thr:raryhidden2}
    Let $x\in\Z_r^n$ be uniformly distributed over a set $A\subseteq \Z_r^n$ of size $|A| \geq r^n/2^c$ for some $c\geq 1$. If $\alpha \leq 1/2$, there is a universal constant $\gamma>0$ (independent of $n$, $t$, $r$ and $\alpha$), such that, for all $\varepsilon > 0$, if $c \leq \frac{\gamma}{r}(\frac{\varepsilon^4}{\alpha})^{1/t} (n/t)^{1-1/t}$, then
    \begin{align*}
        \operatorname*{\mathbb{E}}_{M\sim\mathcal{M}_{t,n}^\alpha}[\|p_M - U\|_{\operatorname{tvd}}] \leq \varepsilon^2.
    \end{align*}
\end{theorem}
\begin{proof}
	Let $f:\Z_r^n\to\{0,1\}$ be the characteristic function of $A$, i.e., $f(x)=1$ iff $x\in A$. We shall bound the Fourier coefficients of $p_M$, which are related to the Fourier coefficients of $f$ as follows:
	\begin{align*}
		\displaybreak[0]\widehat{p}_M(V) = \frac{1}{r^{\alpha n/t}}\sum_{z\in\Z_r^n}p_M(z)\omega_r^{-V\cdot z}\displaybreak[0]
		&= \frac{1}{|A|r^{\alpha n/t}}\sum_{z\in\Z_r^n}|\{x\in A:Mx=z\}|\cdot\omega_r^{-V\cdot z}\displaybreak[0]\\
		&= \frac{1}{|A|r^{\alpha n/t}}\sum_{k=0}^{r-1}|\{x\in A:(Mx)\cdot V=k\}|\cdot\omega_r^{-k}\displaybreak[0]\\
		&= \frac{1}{|A|r^{\alpha n/t}}\sum_{k=0}^{r-1}|\{x\in A:x\cdot (M^{\top}V)=k\}|\cdot\omega_r^{-k}\displaybreak[0]\\
		&= \frac{1}{|A|r^{\alpha n/t}} \sum_{x\in A}\omega_r^{-x\cdot (M^{\top}V)}\displaybreak[0]\\
		&= \frac{r^n}{|A|r^{\alpha n/t}}\widehat{f}(M^{\top}V).
	\end{align*}

	We now start bounding the expected \emph{squared} total variation distance,
	\begin{align*}
		\displaybreak[0]\operatorname*{\mathbb{E}}_{M\sim\mathcal{M}_{t,n}^\alpha}[\|p_M - U\|^2_{\operatorname{tvd}}] &\leq r^{2\alpha n/t}\operatorname*{\mathbb{E}}_{M\sim\mathcal{M}_{t,n}^\alpha}[\|p_M - U\|^2_2]\displaybreak[0] \tag{Cauchy-Schwarz}\\
		&= r^{2\alpha n/t}\operatorname*{\mathbb{E}}_{M\sim\mathcal{M}_{t,n}^\alpha}\left[\sum_{V\in\Z_r^{\alpha n/t}\setminus\{0^{\alpha n/t}\}}|\widehat{p}_M(V)|^2 \right]\displaybreak[0] \tag{Parseval's identity}\\
		&= \frac{r^{2n}}{|A|^2}\operatorname*{\mathbb{E}}_{M\sim\mathcal{M}_{t,n}^\alpha}\left[\sum_{V\in\Z_r^{\alpha n/t}\setminus\{0^{\alpha n/t}\}}|\widehat{f}(M^{\top}V)|^2 \right].
	\end{align*}
	Note that there is at most one $V\in\Z_r^{\alpha n/t}$ such that $S =M^{\top}V$ for a given $S\in\Z_r^n$ (and that the only $V$ that makes $M^{\top}V = 0^n$ is $V=0^{\alpha n/t}$). This allows us to transform the expectation over hypermatchings into a probability,
	\begin{align*}
		\operatorname*{\mathbb{E}}_{M\sim\mathcal{M}_{t,n}^\alpha}[\|p_M - U\|^2_{\operatorname{tvd}}] &\leq \frac{r^{2n}}{|A|^2}\operatorname*{\mathbb{E}}_{M\sim\mathcal{M}_{t,n}^\alpha}\left[\sum_{S\in\Z_r^{n}\setminus\{0^{n}\}}|\{V\in\Z_r^{\alpha n/t}:M^{\top}V=S\}|\cdot|\widehat{f}(S)|^2 \right]\\
		&= \frac{r^{2n}}{|A|^2}\sum_{S\in\Z_r^{n}\setminus\{0^{n}\}}\operatorname*{Pr}_{M\sim\mathcal{M}_{t,n}^\alpha}[\exists V\in\Z_r^{\alpha n/t}:M^{\top}V=S]\cdot|\widehat{f}(S)|^2.
	\end{align*}
	Now observe that $\operatorname*{Pr}_{M\sim\mathcal{M}_{t,n}^\alpha}[\exists V\in\Z_r^{\alpha n/t}:M^{\top}V=S]$ is exactly the probability from \cref{lem:probcomb}, i.e., given $S\in\Z_r^n$ with $k_j \triangleq \frac{1}{t}\cdot |\{i\in[n]:S_i = j\}|\in\Z$ for $j\in[r-1]$ (the number of entries from $S$ equal to $j\neq 0$ must be a multiple of $t$), and defining $k \triangleq \sum_{j=1}^{r-1} k_j$, then 
	\begin{align*}
		\operatorname*{Pr}_{M\sim\mathcal{M}_{t,n}^\alpha}[\exists V\in\Z_r^{\alpha n/t}:M^{\top}V=S] = \frac{\binom{\alpha n/t}{k}}{\binom{n}{kt}}\frac{k!}{(kt)!}\prod_{j=1}^{r-1}\frac{(k_jt)!}{k_j!} \leq \frac{\binom{\alpha n/t}{k}}{\binom{n}{kt}},
	\end{align*}
	and so
	\begin{align*}
		\operatorname*{\mathbb{E}}_{M\sim\mathcal{M}_{t,n}^\alpha}[\|p_M - U\|^2_{\operatorname{tvd}}] &\leq \frac{r^{2n}}{|A|^2}\sum_{k=1}^{\alpha n/t}\frac{\binom{\alpha n/t}{k}}{\binom{n}{kt}}\sum_{\substack{S\in\Z_r^n \\ |S| = kt}}|\widehat{f}(S)|^2.
	\end{align*}	    
    Similarly to the quantum proof, we shall split the above sum into two parts: one in the range $1\leq k < 2rc$, and the other in the range $2rc \leq k \leq \alpha n/t$. 

	\textbf{Sum I} ($1\leq k < 2rc$): in order to upper bound each term, pick $\delta = k/(2rc)$ in \cref{lem:kkl-largefields}, thus
	\begin{align*}
		\frac{r^{2n}}{|A|^2}\sum_{\substack{S\in\Z_r^n \\ |S| = kt}} |\widehat{f}(S)|^2 \leq \frac{r^{2n}}{|A|^2}\frac{1}{\delta^{kt}}\sum_{S\in\Z_r^n}\delta^{|S|} |\widehat{f}(S)|^2 \leq \frac{1}{\delta^{kt}}\left(\frac{r^n}{|A|}\right)^{2r\delta/(1+r\delta)} \leq \frac{1}{\delta^{kt}}\left(\frac{r^n}{|A|}\right)^{2r\delta} \leq \left(\frac{2^{1/t}2rc}{k}\right)^{kt}.
	\end{align*}
	By using that $c \leq \frac{\gamma}{r}(\frac{\varepsilon^4}{\alpha})^{1/t} (n/t)^{1-1/t}$ and $\binom{q}{s}\binom{\ell q}{\ell s}^{-1} \leq (\frac{s}{q})^{(\ell-1)s}$ (see~\cite[Appendix~A.5]{shi2012limits}) for $q=n/t,s=k,\ell=t$, we therefore have
	\begin{align*}
		\frac{r^{2n}}{|A|^2}\sum_{k=1}^{2rc-1}\alpha^k\frac{\binom{n/t}{k}}{\binom{n}{kt}}\sum_{\substack{S\in\Z_r^n \\ |S| = kt}}|\widehat{f}(S)|^2 \leq \sum_{k=1}^{2rc-1}\alpha^k \left(\frac{kt}{n}\right)^{(1-1/t)kt}\left(\frac{2^{1/t}2rc}{k}\right)^{kt} \leq \sum_{k=1}^{2rc-1}\ \left(\frac{2^{1/t}2\gamma \varepsilon^{4/t}}{k^{1/t}}\right)^{kt} \leq \frac{\varepsilon^4}{2},
	\end{align*}
	where we used that $\binom{\alpha n/t}{k} \leq \alpha^k \binom{n/t}{k}$ for $\alpha\in[0,1]$ at the beginning and picked $\gamma$ sufficiently small.
	
	\textbf{Sum II} ($2rc \leq k \leq \alpha n/t$): first note that the function $g(k) \triangleq \binom{\alpha n/t}{k}/\binom{n}{kt}$ is decreasing in the interval $1\leq k \leq \alpha n/t$ (since $\alpha \leq 1/2$). Hence, by using Parseval's identity $\sum_{S\in\Z_r^n} |\widehat{f}(S)|^2 = |A|/r^n$ and the inequality $\binom{q}{s}\binom{\ell q}{\ell s}^{-1} \leq (\frac{s}{q})^{(\ell-1)s}$ (for $q=n/t,s=2m,\ell=t$) in order to bound $g(2rc)$,
	\begin{align*}
		\frac{r^{2n}}{|A|^2}\sum_{k=2rc}^{\alpha n/t}\frac{\binom{\alpha n/t}{k}}{\binom{n}{kt}}\sum_{\substack{S\in\Z_r^n \\ |S| = kt}}|\widehat{f}(S)|^2 &\leq 2^cg(2rc) \leq 2^c\alpha^{2rc}\left(\frac{2rc}{n/t}\right)^{2(t-1)rc}  = 2^c\alpha^{2rc/t}\left(\frac{2\gamma\varepsilon^{4/t}}{(n/t)^{1/t}}\right)^{2(t-1)rc} \leq \frac{\varepsilon^4}{2},
	\end{align*}
	where the last step used that $2(1-1/t)c \geq 1\implies \varepsilon^{2(1-1/t)c} \leq \varepsilon$ and picked $\gamma$ sufficiently small. 
	
	Summing both results, if $c \leq \frac{\gamma}{r}(\frac{ \varepsilon^4}{\alpha})^{1/t} (n/t)^{1-1/t}$, then $\operatorname*{\mathbb{E}}_{M\sim\mathcal{M}_{t,n}^\alpha}[\|p_M - U\|^2_{\operatorname{tvd}}] \leq \varepsilon^4$. By Jensen's inequality, we finally get $\operatorname*{\mathbb{E}}_{M\sim\mathcal{M}_{t,n}^\alpha}[\|p_M - U\|_{\operatorname{tvd}}] \leq \sqrt{\operatorname*{\mathbb{E}}_{M\sim\mathcal{M}_{t,n}^\alpha}[\|p_M - U\|^2_{\operatorname{tvd}}]} \leq \varepsilon^2$.
\end{proof}

\end{document}